\documentclass[11pt,a4paper,twoside]{report}
\pdfoutput=1
\usepackage[ngerman,USenglish]{babel} 

\usepackage{etex}				

\usepackage{times}
\usepackage{graphicx}
\usepackage{graphics}
\usepackage{epsfig}
\usepackage{amssymb}
\usepackage{amsmath}
\usepackage{amsthm}
\usepackage{color}
\usepackage{fancyhdr}

\usepackage[bookmarks,colorlinks,breaklinks]{hyperref}  
\definecolor{dullmagenta}{rgb}{0.4,0,0.4}   
\definecolor{darkblue}{rgb}{0,0,0.4}
\hypersetup{linkcolor=red,citecolor=blue,filecolor=dullmagenta,urlcolor=darkblue} 

\usepackage{url}
\usepackage{tikz}
\usetikzlibrary{chains}
\usetikzlibrary{fit}
\usepackage{pgflibraryarrows}		
\usepackage{pgflibrarysnakes}		
\usepackage{epsfig}

\usepackage{enumerate}  
\usepackage{nicefrac} 
\usepackage{braket}

\usepackage{bm}  
\usepackage{bbm}
\usepackage{xcolor}
\usepackage{verbatim}
\usepackage[backend=bibtex,sorting=none,style=phys,biblabel=brackets,backref=false,bibencoding=utf8,maxnames=20,eprint=true]{biblatex}
\makeatletter
\pretocmd{\blx@head@bibintoc}{\phantomsection}{}{\ddt}
\makeatother
\DefineBibliographyStrings{english}{%
 backrefpage = {page},
 backrefpages = {pages},
}
\usepackage{listings}
\usepackage{url}

\usepackage[latin1]{inputenc}

\usepackage{paralist} 

\usepackage{mathabx} 

\usepackage{footmisc} 

 \newcommand{\Mat}{\textnormal{Mat}}
 \newcommand{\s}{\textnormal{s}}
  \newcommand{\diag}{ \textnormal{diag}}

\newcommand{\ketbra}[1]{|#1\rangle\langle #1|}

\newcommand*{\ee}{\mathrm{e}}

\newcommand{\id}{\textnormal{id}}

\def\tr{{\rm tr}}

\def\Re{{\rm Re}}
\def\cl{{\text{cl}}}
\def\ker{{\text{ker}}}
\def\fpg{F_{{\rm pg}}}

\def\cD{\mathcal D}
\def\cE{\mathcal E}
\def\spec{\text{spec}}

\newcommand{\norm}[1]{\left\lVert#1\right\rVert}
\newcommand{\normT}[1]{\left\vert\kern-0.25ex\left\vert\kern-0.25ex\left\vert #1 
    \right\vert\kern-0.25ex\right\vert\kern-0.25ex\right\vert}

\newtheorem{mythm}{Theorem}[section]
\newtheorem{myprop}[mythm]{Proposition}
\newtheorem{mycor}[mythm]{Corollary}
\newtheorem{mylem}[mythm]{Lemma}

\theoremstyle{definition}
\newtheorem{mydef}[mythm]{Definition}

\newtheorem{myrmk}[mythm]{Remark}

\numberwithin{equation}{chapter}
\fussy
\large \normalsize
\begin{document}

%
%
\selectlanguage{USenglish}
%
\pagenumbering{Roman}
%
%
%
%
%
%
\begin{titlepage}
  \mbox{}

  \vspace{-1.5cm}
  \noindent
  \begin{tabular}{@{} l @{} l @{}}
    \begin{minipage}[c]{0.5\textwidth}
      \hspace{-4mm}
      \includegraphics[height=19mm]{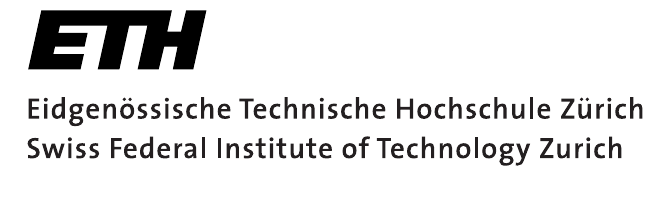}
    \end{minipage} &
    \begin{minipage}[c]{0.5\textwidth}
       \hfill \vspace{-6 mm}\includegraphics[height=8mm]{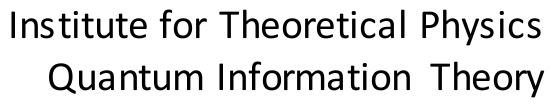}
    \end{minipage} \\
  \end{tabular}
  \rule{\textwidth}{0.5pt}
  \begin{center}
    {\Large 
      Fall 2016 \hfill Prof.~Dr.~Renato Renner
    }
    
    \vspace{\stretch{5}}
    \LARGE
    Master's Thesis
 
    \vspace{\stretch{8}}
    \Huge\textbf{
Relations between different quantum R\'enyi divergences
          }
    
    \vspace{\stretch{10}}
    \LARGE{
      Raban Iten
    }
    
    \vspace{\stretch{10}}
    \rule{\textwidth}{0.5pt}
   
    \vspace{0.0cm}
    \begin{flushleft}
      \begin{tabular}{ll}
        \Large Advisors: & \Large 
        David~Sutter
        \\
        \Large   & \Large
        Dr.~Joseph Merrill Renes 
        \\
        \Large & \Large 
        Prof.~Dr.~Renato Renner
      \end{tabular}
    \end{flushleft}
  \end{center}
\end{titlepage}
%

\thispagestyle{plain}
\cleardoublepage

%
%
\chapter*{Acknowledgments}
I would like to deeply thank my advisors David Sutter, Dr. Joseph M. Renes and Prof. Renato Renner for their excellent supervision. They were always available to discuss problems and offered an outstanding support. None of the achieved results in this thesis would have been possible without them. I would also like to thank them for coming up with the fascinating project tasks and for their great preparatory works on these topics.\\

Furthermore, I want to thank Roger Colbeck for his mathematica package QItools, which we used several times to check our analytical conjectures numerically. 

\vspace{1.1cm}
\noindent
Zurich,  Dezember 18, 2016

\vspace{2.4cm}
\noindent
Raban Iten
\thispagestyle{plain}
\clearpage

\thispagestyle{plain}
\cleardoublepage
%
%
\huge
\begin{abstract}
  \setcounter{page}{5}
\thispagestyle{plain}
  \normalsize
  \vspace{0.5cm}

\noindent Quantum generalizations of R\'enyi's entropies are a useful tool to describe a variety of operational tasks in quantum information processing. 
Two families of such generalizations turn out to be particularly useful: the Petz quantum R\'enyi divergence~$\widebar{D}_{\alpha}$ and the minimal quantum R\'enyi divergence~$\widetilde{D}_{\alpha}$. Moreover, the maximum quantum R\'enyi divergence~$\widehat{D}_{\alpha}$ is of particular mathematical interest. In this thesis, we investigate relations between these divergences and their applications in quantum information theory. As the names suggest, it is well known that $\widetilde{D}_{\alpha}(\rho \| \sigma) \leqslant \widebar{D}_{\alpha}(\rho \| \sigma)\leqslant \widehat{D}_{\alpha}(\rho \| \sigma)$ for $\alpha \geqslant 0$ and where $\rho$ and $\sigma$ are density operators.\\

Our main result is a reverse Araki-Lieb-Thirring inequality that implies a new and reverse relation between the minimal and the Petz divergence, namely that $\alpha \widebar{D}_{\alpha}(\rho \| \sigma) \leqslant  \widetilde{D}_{\alpha}(\rho \| \sigma)$ for $\alpha \in [0,1]$. This bound leads to a unified picture of the relationship between pretty good quantities used in quantum information theory and their optimal versions. 
Indeed, the bound suggests defining a ``pretty good fidelity'', whose relation to the usual fidelity implies the known relations between the optimal and pretty good measurement as well as the optimal and pretty good singlet fraction.
We also find a new necessary and sufficient condition for optimality of the pretty good measurement and singlet fraction.\\

In addition, we provide a new proof of the inequality $\widetilde{D}_{1}(\rho \| \sigma) \leqslant  \widehat{D}_{1}(\rho \| \sigma)\, ,$ based on the Araki-Lieb-Thirring inequality. This leads to an elegant proof of the logarithmic form of the reverse Golden-Thompson inequality.
\vspace{3mm}

\noindent 

\vspace{10mm}
\noindent{\bf{Keywords}}
Reverse Araki-Lieb-Thirring inequality, reverse Golden-Thompson inequality, R\'enyi divergences, R\'enyi entropies, optimality of pretty good measures, pretty good measurement, pretty good singlet fraction 
\end{abstract}

\normalsize

\selectlanguage{USenglish}
\thispagestyle{plain}
\cleardoublepage
%
%

%
\tableofcontents
\clearpage

\thispagestyle{plain}
\cleardoublepage

\lhead[\fancyplain{\scshape \leftmark}
{\scshape \leftmark}]
{\fancyplain{\scshape Semester Project}
  {\scshape Semester Project}}
\rhead[\fancyplain{\scshape Semester Project}
{\scshape Semester Project}]
{\fancyplain{\scshape \leftmark}
  {\scshape \leftmark}}
%
%
%
%
%
\setcounter{chapter}{0}
\setcounter{figure}{0}
%
%
\pagenumbering{arabic}
\renewcommand{\thechapter}{\arabic{chapter}}
\renewcommand{\thesection}{\thechapter.\arabic{section}}
\renewcommand{\thefigure}{\thechapter.\arabic{figure}}
\renewcommand{\chaptermark}[1]{\markboth{#1}{}}
\renewcommand{\sectionmark}[1]{\markright{\thesection\ #1}}
\lhead[\fancyplain{\scshape Chapter \thechapter}
{\scshape Chapter \thechapter}]
{\fancyplain{\textsc{\leftmark}}
  {\rightmark}}
\rhead[\fancyplain{\scshape \leftmark}
{\textsc{\leftmark}}]
{\fancyplain{\scshape Chapter \thechapter}
  {\scshape Chapter \thechapter}}
%
%
\cleardoublepage

\chapter{Preface} 

How can we quantify information? How can we measure the uncertainty about a physical system? These questions are not easy to answer and there are different useful measures which provide possible solutions to these questions. The problem gets even more complex if we consider quantum systems instead of classical ones. It turns out that quantum R\'enyi divergences provide a useful framework to deal with such questions. Interesting distance measures between two quantum states such as the fidelity are nicely embedded into this framework.  \\

A natural question is how different information measures are related to each other. In this thesis, we introduce a new relation between two families of such measures and describe its applications. Let us give an example of one such application: Assume that Alice prepares a certain quantum state $\rho$ with probability $p$ and a state $\sigma$ with probability $1-p$. The state is then given to Bob, who knows that Alice prepared either $\rho$ or $\sigma$, but does not know which one of both was prepared. Which measurement should Bob perform to find out if Alice has given him $\rho$ or $\sigma$ with the highest possible success probability? Unfortunately, this problem is not easy to solve in general. However, there is a known construction of a "pretty good" measurement, which provides a pretty good solution to this problem. (We refer to Appendix~\ref{app:pgm} for more details.) The relations between different quantum R\'enyi divergences allow us to specify what "pretty good" means mathematically (by comparing the "pretty good" measure with the optimal one) and to give necessary and sufficient conditions on the optimality of the pretty good measurement (cf. Chapter~\ref{cha:pgm} for more details).\\

\vspace{12mm}

The thesis is structured as follows. In Chapter~\ref{Divergences}, we give some background information about quantum R\'enyi divergences, entropies and introduce a natural continuation of the important minimal quantum R\'enyi divergence (also known as sandwiched quantum R\'enyi divergence) $\widetilde{D}_{\alpha}$ for $\alpha \in (0,\frac{1}{2})$. \\
In Chapter~\ref{chap:trace_ineq}, we consider several trace inequalities that are not only of mathematical interest, but also found a lot of applications in quantum information theory. Our main result is a reverse version of the celebrated Araki-Lieb-Thirring (ALT)  inequality. Moreover, we give a new and elegant proof (based an the ALT inequality) of a logarithmic trace inequality which is known to be equivalent to the reverse Golden-Thompson inequality.\\
In Chapter~\ref{chap:rel_div}, we introduce a new bound between two well known quantum R\'enyi divergences, the minimal  quantum R\'enyi divergence  $\widetilde{D}_{\alpha}$ and the Petz quantum R\'enyi divergence $ \widebar{D}_{\alpha}$, namely that $\alpha \widebar{D}_{\alpha}(\rho \| \sigma) \leqslant  \widetilde{D}_{\alpha}(\rho \| \sigma)$ for $\alpha \in [0,1]$ and for density operators $\rho$ and $\sigma$. This bound is a direct consequence of the reverse ALT inequality derived in Chapter~\ref{chap:trace_ineq} and leads to interesting new bounds between quantum conditional  R\'enyi entropies.\\
In Chapter~\ref{cha:pgm}, we describe applications of the new bound between the minimal and the Petz quantum R\'enyi divergence found in Chapter~\ref{chap:rel_div}. Indeed, the bound turns out to be useful to quantify the quality of different "pretty good" measures in quantum information theory and provides a unification of known bounds for such measures.\\

In Appendix~\ref{app:pgm}, we give a formal description of the pretty good measurement. \\
In Appendix~\ref{app:technical_res}, we prove some technical results related to statements araising in the main text of the thesis.\\
Appendix~\ref{APPnoation} explains the notational conventions and abbreviations we use.\\

The new bound  between the minimal and the Petz quantum R\'enyi divergence given in Chapter~\ref{chap:rel_div} as well as its applications described in Chapter~\ref{cha:pgm} have been summarized in a paper~\cite{PGM} together with David Sutter and Dr.~Joseph Merrill Renes. The paper was recently accepted for a publication in~\emph{IEEE Transaction on Information Theory}. 

\vspace{12mm}

\chapter{Quantum R\'enyi divergences} \label{Divergences}

\section{Introduction} \label{sec:intro_divergences}
As with their classical counterparts, quantum generalizations of R\'enyi entropies and divergences are powerful tools in information theory. They are related to various measures of information and uncertainty, which are useful for different tasks in finite resource theory. The aim of finite resource theory is to understand the information processing of a finite amount of resources, e.g., channels. A nice example that illustrates the usefulness of classical R\'enyi entropies for the investigation of source compression is given in Section~1.1 of~\cite{tomamichel_quantum_2016}. \\

Alfr{\'e}d R{\'e}nyi derived an elegant axiomatic approach for classical R\'enyi entropies and divergences~\cite{renyi_measures_1961}. Indeed, he states five natural\footnote{The axioms for entropies or divergences describe desirable properties of uncertainty measures or measures of distinguishability, respectively.} axioms on functionals on a probability space that allow only one solution: the well known Shannon entropy~\cite{shannon_mathematical_1948} or the Kullback-Leibler divergence~\cite{kullback_information_1951}, respectively. Classical probability distributions can be viewed as diagonal operators $\rho \neq 0$ and $\sigma \gg \rho$ (with unit trace), where the notation $\sigma \gg \rho$ denotes that the kernel of $\sigma$ is a subset of the kernel of $\rho$. Then, the (classical) Kullback-Leibler divergence is defined as 

\begin{equation}  \label{eq:rel_entr}
D(\rho \| \sigma):=\frac{ \tr \, \rho ( \log  \rho  - \log  \sigma) }{\tr \, \rho} \, .
\end{equation}
To ensure continuity, we use the convention that  $0 \log 0 = 0$. \\
Relaxing one of the five axioms allows then (in addition to the  Kullback-Leibler divergence) a whole family of divergences, the so called R\'enyi divergences. For $\alpha \in  (0,1) \cup  (1,\infty)$ and diagonal operators $\rho \neq 0$ and $\sigma \gg \rho$, the (classical) R\'enyi  divergences are defined as

\begin{equation}  \label{eq:rel_reny_entropy}
D_{\alpha}(\rho \| \sigma):=\frac{1}{\alpha-1} \log \frac{ \tr \, \rho^{\alpha} \sigma^{1-\alpha} }{\tr \, \rho} \, .
\end{equation}

To ensure continuity, we use the convention that $\tfrac{0}{0} = 1$. \\

Adapting R{\'e}nyi's axioms to the quantum case leads to the following axioms~\eqref{enum:axiom1}-\eqref{enum:axiom6}. Let  $\rho \neq 0$,  $\tilde{\rho} \neq 0$ and $\sigma \gg \rho$, $\tilde{\sigma} \gg \tilde{\rho}$ be non-negative operators. Then,  a quantum R\'enyi  divergence $\mathbb{D}$ satisfies all of the following axioms. (We refer to~\cite{tomamichel_quantum_2016} for a more detailed discussion of the axioms.)

\begin{enumerate}[(I)] 
  \item \textbf{Continuity}: $\mathbb{D}(\rho \| \sigma)$ is continuous in $\rho$ and $\sigma$.\footnote{Note that this axiom excludes quantum R\'enyi divergences with parameters $\alpha \leqslant 0$.} \label{enum:axiom1}
  \item\textbf{ Unitary invariance}:  $\mathbb{D}(\rho \| \sigma)= \mathbb{D}(U\rho U^{\dagger} \| U\sigma U^{\dagger})$ for any unitary $U$. \label{axiom:Uni}
  \item\textbf{ Normalization}: $\mathbb{D}(1 \| \frac{1}{2})=\log(2)$.
  \item \textbf{Order}: If $\rho\geqslant \sigma$, then $\mathbb{D}(\rho \| \sigma)\geqslant 0$. If $\rho \leqslant \sigma $, then $\mathbb{D}(\rho \| \sigma)\leqslant 0$.
  \item \textbf{Additivity}:  $\mathbb{D}(\rho \otimes \tilde{\rho} \| \sigma \otimes \tilde{\sigma})=\mathbb{D}(\rho \| \sigma)+\mathbb{D}(\tilde{\rho} \| \tilde{\sigma})$. \label{axiom:Add}
  \item \textbf{General Mean}: There exists a continuous and strictly monotonic function $g$ such that $\mathbb{Q}(\cdot \| \cdot) := g(\mathbb{D}(\cdot \| \cdot))$ satisfies 
  
\begin{equation}  \label{eq:gen_mean} 
\mathbb{Q}(\rho \oplus \tilde{\rho} \| \sigma \oplus \tilde{\sigma})=\frac{\tr \rho}{\tr (\rho+\tilde{\rho})} \mathbb{Q}(\rho\| \sigma) +  \frac{\tr \tilde{\rho}}{\tr (\rho+\tilde{\rho})} \mathbb{Q}(\tilde{\rho}\| \tilde{\sigma}) \, .
\end{equation}
\label{enum:axiom6}
 \end{enumerate}

Since it is desirable to have an interpretation of a R\'enyi  divergence as a measure of distinguishability, the following property is desirable. 
\begin{enumerate}[(DPI)] 
  \item  \textbf{Data-processing inequality}: For all completely positive, trace-preserving (CPTP) maps $\mathcal{E}$ and for all non-negative operators $\rho$ and $\sigma$, we have
\begin{equation}\label{eq:DPI}
\mathbb{D}(\rho \| \sigma)\geqslant \mathbb{D}\big(\mathcal{E}(\rho) \|\mathcal{E}( \sigma)\big) \, .
\end{equation}
  \end{enumerate}

The DPI can be viewed as the statement that the distinguishability of two density operators $\rho$ and $\sigma$ can only decrease under the application of a quantum channel. In the classical case, the axioms~\eqref{enum:axiom1}-\eqref{enum:axiom6} imply the data-processing inequality (DPI), but it is an open question if this is also the case in the quantum case. Note that the DPI is mathematically more involved than the axioms~\eqref{enum:axiom1}-\eqref{enum:axiom6}.\\

In contrast to the classical case, there is not a unique family of functionals that satisfies the axioms~\eqref{enum:axiom1}-\eqref{enum:axiom6} in the quantum case. This is based on the non commuting nature of quantum mechanics.  Indeed, since two non-negative operators do not commute in general, the order of the operators in the functional matters in the quantum case and leads to more possibilities than in the classical one. Interestingly, there are several different functionals (so called quantum R\'enyi  divergences) which satisfy~\eqref{enum:axiom1}-\eqref{enum:axiom6} and the DPI.

In the following, we restrict our attention to four families of quantum  R\'enyi divergences. The two most important ones are the \emph{Petz quantum R\'enyi divergence}~\cite{petz_quasi-entropies_1986} and the \emph{minimal quantum R\'enyi divergence}~\cite{muller-lennert_quantum_2013,wilde_strong_2014} (also known as \emph{sandwiched quantum R\'enyi divergence}), which have proven particularly useful, finding application to achievability, strong converses, and refined asymptotic analysis of a variety of coding and hypothesis testing problems (for a recent overview, see~\cite{tomamichel_quantum_2016}). \\
The \emph{reverse minimal quantum R\'enyi divergence} was introduced in   ~\cite{audenaert_alpha-z-relative_2015} under the name "reverse sandwiched R\'enyi relative entropy". It is especially interesting in the limit $\alpha \rightarrow 0$, where it reduces to the 0-R\'enyi relative divergence, which has been used for one-shot information theory~\cite{wang_one-shot_2012, buscemi_entanglement_2011}.\\
In addition, we consider the \emph{maximal quantum R\'enyi divergence}~\cite{Matsumoto_2016,tomamichel_quantum_2016}, which is mathematically interesting, since it provides an upper bound on all possible quantum R\'enyi divergences. Moreover, it is related to the geometric mean of two matrices.\\

\section{Quantum R\'enyi divergences} \label{sec:important_divergences}
We define four important families of quantum R\'enyi divergences. An overview over the different families of divergences for fixed arguments $\rho$ and $\sigma$ is shown in Figure~\ref{fig: divergences}.

For two non-negative operators $\rho\neq 0$ and $\sigma$ and $\alpha \in (0,1)\cup (1,\infty)$, the \emph{Petz quantum R\'enyi divergence} is defined as
\begin{equation} \label{eq:petz_entropy}
\widebar{D}_{\alpha}(\rho \| \sigma):= \begin{cases}
\frac{1}{\alpha-1} \log \frac{1}{\tr \rho} \widebar{Q}_{\alpha}(\rho \| \sigma) &\text{if $ \sigma \gg \rho \lor \alpha<1 $}\\
\infty &\text{otherwise}\, ,
\end{cases}
\end{equation}
where $\widebar{Q}_{\alpha}(\rho\|\sigma):=   \tr \rho^\alpha \sigma^{1-\alpha}$ and we use the common convention that  $-\log 0=\infty$. Moreover, negative matrix powers are only evaluated on the support of the non-negative operator throughout this thesis. The Petz divergence $\widebar{D}_{\alpha}$ satisfies the axioms~\eqref{enum:axiom1}-\eqref{enum:axiom6} and the DPI for $\alpha \in (0,1) \cup (1,2]$~\cite{petz_quasi-entropies_1986}.

\begin{figure}[!htb]
\centering
 \includegraphics[height=50mm]{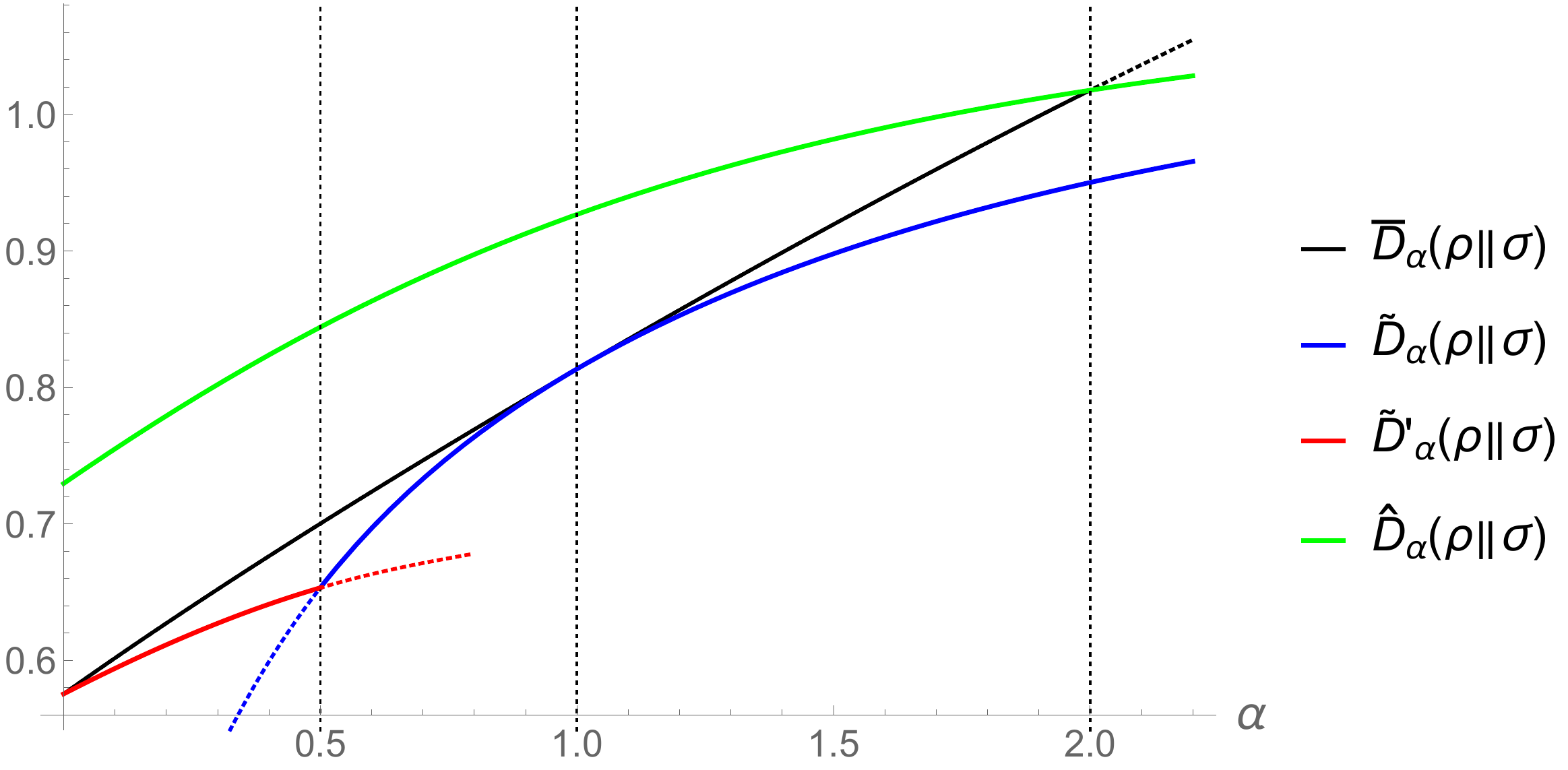}
\caption[Overview of divergence families]{Important families of divergences represented dependent on the parameter $\alpha$ for fixed density matrices $\rho$ and $\sigma$ (where we have chosen the matrices $\rho$ and $\sigma$ as in Figure~4.1of~\cite{tomamichel_quantum_2016}, for ease of comparison). The dashed line signalizes that the DPI is not satisfied for the corresponding $\alpha$.}
\label{fig: divergences}
\end{figure}

The \emph{minimal quantum R\'enyi divergence} (which is also called "sandwiched quantum R\'enyi relative entropy") on the other hand is defined by 
\begin{equation}  \label{eq:min_entropy} 
\widetilde{D}_{\alpha}(\rho \| \sigma):= \begin{cases}
\frac{1}{\alpha-1} \log  \frac{1}{\tr \rho}  \widetilde{Q}_{\alpha}(\rho \| \sigma) &\text{if $ \sigma \gg \rho \lor \alpha<1 $}\\
\infty &\text{otherwise} \, ,
\end{cases}
\end{equation}
where $\widetilde{Q}_{\alpha}(\rho\|\sigma):=\tr \left(\sigma^{\frac{1-\alpha}{2\alpha}}\rho\sigma^{\frac{1-\alpha}{2\alpha}}\right)^\alpha$. The minimal divergence $\widetilde{D}_{\alpha}$  satisfies the axioms~\eqref{enum:axiom1}-\eqref{enum:axiom6} and the DPI for $\alpha \in [\frac{1}{2}, \infty]$~\cite{frank_monotonicity_2013} (see also~\cite{beigi_sandwiched_2013}).

We will show in Section~\ref{sec:min_max_divergence} that the natural continuation of the minimal quantum R\'enyi divergence for $\alpha \in (0,\frac{1}{2})$ that satisfies the DPI is given by the \emph{reverse minimal  quantum R\'enyi divergence}, which is defines as follows 
\begin{equation}  \label{eq:min_reversed_entropy} 
\widetilde{D'}_{\alpha}(\rho \| \sigma):= \frac{1}{\alpha-1} \log  \frac{1}{\tr \rho}  \widetilde{Q'}_{\alpha}(\rho \| \sigma)  \, ,
\end{equation}
where $\widetilde{Q'}_{\alpha}(\rho\|\sigma):=\tr \left(\rho^{\frac{\alpha}{2(1-\alpha)}}\sigma\rho^{\frac{\alpha}{2(1-\alpha)}}\right)^{1-\alpha}$ for $\alpha <1$ and non-negative operators $\rho\neq 0$ and $\sigma$. This divergence was introduced in a different context in~\cite{audenaert_alpha-z-relative_2015} under the name "reverse sandwiched R\'enyi relative entropy", where it was also shown that it satisfies the axioms~\eqref{enum:axiom1}-\eqref{enum:axiom6} and the DPI~for $\alpha \in (0,\frac{1}{2}]$. Note that the name of the reverse minimal  quantum R\'enyi divergence is motivated by the following symmetry relation introduced in equation~(10) in~\cite{audenaert_alpha-z-relative_2015}: For any density operators $\rho\neq0$ and $\sigma \gg \rho \,$, we have that
\begin{equation}  \label{eq:sym_rel} 
\widetilde{D'}_{\alpha}(\rho \| \sigma)=\frac{\alpha}{1-\alpha} \widetilde{D}_{1-\alpha}(\sigma \| \rho  ).
\end{equation}

The \emph{maximal quantum R\'enyi divergence} on the other hand is defined by 
\begin{equation}  \label{eq:max_entropy}
\widehat{D}_{\alpha}(\rho \| \sigma):= \begin{cases}
\frac{1}{\alpha-1} \log  \frac{1}{\tr \rho}  \widehat{Q}_{\alpha}(\rho \| \sigma) &\text{if $ \sigma \gg \rho $}\\
\infty &\text{otherwise} \, ,
\end{cases}
\end{equation}
where $\widehat{Q}_{\alpha}(\rho\|\sigma):=\tr \, \sigma \left(\sigma^{-\frac{1}{2}}\rho\sigma^{-\frac{1}{2}}\right)^\alpha$. For $\alpha < 1$, this expression is the trace of a matrix mean $\widehat{Q}_{\alpha}(\rho\|\sigma)=\tr \, \sigma^{\frac{1}{2}} \left(\sigma^{-\frac{1}{2}}\rho\sigma^{-\frac{1}{2}}\right)^\alpha\sigma^{\frac{1}{2}}=:\tr \, \sigma \#_{\alpha} \rho$~\cite{kubo_means_1979} . In particular, $\sigma \#_{\alpha} \rho$ corresponds to the geometric mean of $\rho$ and $\sigma$ for $\alpha=\nicefrac{1}{2}$. The joint concavity of the matrix means $ \sigma \#_{\alpha} \rho$ leads then to the DPI for  $\widehat{D}_{\alpha}$ (cf. for example~\cite{tomamichel_quantum_2016}  for more details).

Moreover, we define $\mathbb{D}_{0}$, $\mathbb{D}_{1}$ and $\mathbb{D}_{\infty}$  as limits of $\mathbb{D}_{\alpha}$ for $\alpha \rightarrow 0$, $\alpha \rightarrow 1$ and $\alpha \rightarrow \infty$, respectively, for any family of quantum R\'enyi divergences  $\mathbb{D}_{\alpha}$.\\

\begin{myrmk}
By construction, all quantum  R\'enyi divergences $\mathbb{D}_{\alpha}(\rho \| \sigma)$ reduces to the corresponding classical R\'enyi divergence for non-negative diagonal operators $\rho$ and $\sigma$. In other words, we have that $\mathbb{D}_{\alpha}(\rho \| \sigma)=D_{\alpha}(\rho \| \sigma)$ for all $\alpha \in  (0,1) \cup  (1,\infty)$ and that $\mathbb{D}_{1}(\rho \| \sigma)=D(\rho \| \sigma)$ for  non-negative operators $\rho$ and $\sigma$ with $[\rho,\sigma]=0$, where $D_{\alpha}$ is given in~\eqref{eq:rel_reny_entropy} and $D$ is given in~\eqref{eq:rel_entr}.
\end{myrmk}


\section{Limits of quantum R\'enyi divergences} \label{sec:limit}
 
In this thesis, we are only interested in the limit cases of quantum  R\'enyi divergences for $\alpha \rightarrow 1$, and we refer to~\cite{tomamichel_quantum_2016} for a discussion of the limits $\alpha \rightarrow \infty$ and $\alpha \rightarrow 0$. The divergences $\widetilde{D}_{\alpha}$, $\widebar{D}_{\alpha}$ and $\widehat{D}_{\alpha}$ converge to interesting quantities in the limit $\alpha \rightarrow 1$. It is well known that $\widetilde{D}_{1}=\widebar{D}_{1}$. The derivation of the expressions for  $\widetilde{D}_{1}$ and $\widehat{D}_{1}$ follow quite directly by an application of the l'H\^{o}pital rule. We refer to~\cite{tomamichel_quantum_2016} for the proofs of the following propositions.
 
 \begin{myprop} [Proposition~4.5 in~\cite{tomamichel_quantum_2016}] \label{prop:limit_min}
  Let  $\rho \neq 0$ and $\sigma \gg \rho$. Then
 \begin{equation}  \label{eq:limit_min_div} 
\lim_{\alpha \rightarrow 1} \widetilde{D}_{\alpha}(\rho \| \sigma)=\frac{1}{\tr \, \rho} \tr \,  \rho ( \log  \rho  - \log  \sigma)  \, ,
\end{equation}
where the right hand side corresponds to the quantum divergence (which reduces to the (classica)  Kullback-Leibler divergence given in~\eqref{eq:rel_entr} in the commuting case).
 \end{myprop}
 
 \begin{myprop} [See Section~4.2.3 in~\cite{tomamichel_quantum_2016}] \label{prop:limit_max}
  Let  $\rho \neq 0$ and $\sigma \gg \rho$. Then
 \begin{equation}  \label{eq:limit_max_div} 
\lim_{\alpha \rightarrow 1} \widehat{D}_{\alpha}(\rho \| \sigma)=  \frac{1}{\tr \rho}  \tr \, \rho \log \left( \rho^{\frac{1}{2}} \sigma^{-1}\rho^{\frac{1}{2}} \right) \, ,
\end{equation}
where the expression on the right hand side is known under the name Belavkin-Staszewski relative entropy~\cite{Belvakin1982}.
 \end{myprop}

\section{Minimal and maximal quantum R\'enyi  divergence} \label{sec:min_max_divergence}

In this section, we give lower and upper bounds on arbitrary quantum R\'enyi divergences, where we focus on the lower bound, which has turned out to be useful for many applications in quantum information theory. Using the construction of~\cite{Matsumoto_2016}, it was shown in~\cite{tomamichel_quantum_2016} that every  quantum R\'enyi divergence $\mathbb{D}_{\alpha}$ that satisfies the DPI is smaller than the maximal quantum R\'enyi divergence, i.e., $\mathbb{D}_{\alpha}(\rho \| \sigma) \leqslant \widehat{D}_{\alpha}(\rho \| \sigma)$ for any non-negative operators $\rho$ and $\sigma$ and any $\alpha \geqslant 0$. \\

In~\cite{tomamichel_quantum_2016}, it is shown that $\widetilde{D}_{\alpha}$ provides a lower bound on arbitrary quantum R\'enyi divergences $\mathbb{D}_{\alpha}$. Since  $\widetilde{D}_{\alpha}$ satisfies the DPI for $\alpha \in [\frac{1}{2}, \infty]$, we conclude that $\widetilde{D}_{\alpha}$ is the smallest quantum R\'enyi divergence in this $\alpha$-range. In the following, we show how to find the smallest quantum R\'enyi divergence for $\alpha \in (0,\frac{1}{2})$ that satisfies the DPI using the same proof techniques as used in~\cite{tomamichel_quantum_2016} for the case $\alpha \in [\frac{1}{2}, \infty]$. \\

Let us first recall an interesting characterization of the minimal quantum R\'enyi divergence. For this purpose, we  define the pinching map $\mathcal{P}_{H}$ corresponding to a Hermitian operator $H$. Every Hermitian operator $H$ can be decomposed into  $H=\sum_{\lambda} \lambda P_{\lambda}$, where $\lambda \in  \spec(H)$ are the eigenvalues of $H$ (without multiplicity) and $P_{\lambda}$ are mutually orthogonal projectors. Then, the pinching map $\mathcal{P}_{H}$ is defined as a superoperator on linear operators $L$ by sending $\mathcal{P}_{H}:L\rightarrow \sum_{\lambda} P_{\lambda} L  P_{\lambda}$. Note that $\mathcal{P}_{H}$ is a CPTP, unital and self-adjoint map, which can be viewed as a dephasing operation that remove off-diagonal blocks of a matrix. Clearly, we have that $[\mathcal{P}_H[P],H]=0$ for a non-negative operator $P$. 

\begin{myprop}[Proposition 4.4 in~\cite{tomamichel_quantum_2016}] \label{prop:char_min_div}
Let $\alpha \geqslant 0$ and $\rho\neq 0$ and $\sigma \gg \rho$  be two non-negative operators. Then
\begin{equation}  \label{eq:char_min_div} 
\widetilde{D}_{\alpha}(\rho \| \sigma)= \lim_{n \rightarrow \infty} \frac{1}{n} D_{\alpha}(\mathcal{P}_{\sigma^{\otimes n}}\left[ \rho^{\otimes n} \right] \| \sigma^{\otimes n}) \, ,
\end{equation}
where $D_{\alpha}$ denotes the classical R\'enyi divergence.\footnote{Note that $\sigma^{\otimes n}$ and $\mathcal{P}_{\sigma^{\otimes n}}\left[ \rho^{\otimes n} \right] $ are diagonal in the eigenbasis of $\sigma^{\otimes n}$, which ensures that $D_{\alpha}$ appearing on the right hand side of~\eqref{eq:char_min_div} can be considered to be classical.}
\end{myprop}

Proposition~\ref{prop:char_min_div} leads directly to the minimization property of $\widetilde{D}_{\alpha}$.

\begin{mylem}[Section 4.2.2 in~\cite{tomamichel_quantum_2016}] \label{lem:min_of_min_div}
Let $\rho\neq 0$ and $\sigma \gg \rho$  be two non-negative operators. And let $\mathbb{D}_{\alpha}$ be an arbitrary family of quantum R\'enyi divergences that satisfies the DPI. Then
\begin{equation}  
\widetilde{D}_{\alpha}(\rho \| \sigma)\leqslant \mathbb{D}_{\alpha}(\rho \| \sigma) \qquad \text{ for  $\alpha \geqslant 0$}\, .
\end{equation}
\end{mylem}
\begin{proof}
By the DPI~\eqref{eq:DPI} for an arbitrary family of quantum R\'enyi divergences $\mathbb{D}_{\alpha}$, we find that for any non-negative operators $\rho\neq0$ and $\sigma \gg \rho$ 
\begin{equation}  \label{eq:DPI_min_div} 
\mathbb{D}_{\alpha}(\rho \| \sigma)=\frac{1}{n}\mathbb{D}_{\alpha}(\rho^{\otimes n} \| \sigma^{\otimes n})\geqslant \frac{1}{n}\mathbb{D}_{\alpha}(\mathcal{P}_{\sigma^{\otimes n}}\left[ \rho^{\otimes n} \right] \| \sigma^{\otimes n}) \, ,
\end{equation}
where we used the Additivity property~\eqref{axiom:Add} of R\'enyi divergences in the first equality. Noting that we can replace $\mathbb{D}_{\alpha}$ by the classical R\'enyi divergence  $D_{\alpha}$ on the right hand side of~\eqref{eq:DPI_min_div} , we find the statement of Lemma~\ref{lem:min_of_min_div} by taking the limit $n\rightarrow \infty$ and applying Proposition~\ref{prop:char_min_div}.
\end{proof}

In the following, we improve the lower bound given in Lemma~\ref{lem:min_of_min_div} for $\alpha \in (0,\frac{1}{2}]$. Essentially, the idea for the proof of the following lemma is to interchange the roles of $\rho$ and $\sigma$ in the proof  of Lemma~\ref{lem:min_of_min_div}. This leads very naturally to the form of the reverse minimal quantum R\'enyi divergence and motivates its introduction.

\begin{mylem} \label{lem:min_of_min_div2}
Let $\rho\neq 0$ and $\sigma \gg \rho$  be two non-negative operators. And let $\mathbb{D}_{\alpha}$ be an arbitrary family of quantum R\'enyi divergences that satisfies the DPI. Then
\begin{equation} 
\widetilde{D'}_{\alpha}(\rho \| \sigma)\leqslant \mathbb{D}_{\alpha}(\rho \| \sigma) \qquad \text{ for  $\alpha \in (0,1)$}\, .
\end{equation}
\end{mylem}

We conclude that the reverse minimal quantum R\'enyi divergence is the smallest quantum  R\'enyi divergence that satisfies the axioms~\eqref{enum:axiom1}-\eqref{enum:axiom6} and the DPI for $\alpha \in (0,\frac{1}{2}]$. Moreover, since $\widetilde{D'}_{\frac{1}{2}}(\rho \| \sigma)= \widetilde{D}_{\frac{1}{2}}(\rho \| \sigma)$ for any  non-negative operators $\rho\neq 0$ and $\sigma \gg \rho$, Lemma~\ref{lem:min_of_min_div2} suggest a natural continuation of $\widetilde{D}_{\alpha}(\rho \| \sigma)$ for $\alpha \in (0,\frac{1}{2}]$ (cf. also Figure~\ref{fig: divergences}).

\begin{proof}[Proof of Lemma~\ref{lem:min_of_min_div2}]
Let  $\alpha \in (0,1)$ and let us assume (for the moment) that  $\rho\neq0$ and $\sigma$ are non-negative density operators with $\sigma \gg \rho$ and also $\rho \gg \sigma$. By the DPI, we find that 
\begin{align} 
\mathbb{D}_{\alpha}(\rho \| \sigma)=\frac{1}{n}\mathbb{D}_{\alpha}(\rho^{\otimes n} \| \sigma^{\otimes n})
& \geqslant \frac{1}{n}\mathbb{D}_{\alpha}( \rho^{\otimes n} \| \mathcal{P}_{\rho^{\otimes n}}\left[\sigma^{\otimes n} \right]) \\
&=\frac{1}{n} D_{\alpha}( \rho^{\otimes n} \| \mathcal{P}_{\rho^{\otimes n}}\left[\sigma^{\otimes n} \right])\, .
\end{align}
By a simple substitution $\tilde{\alpha}:=1-\alpha \in (0,1)$, we then find
\begin{equation}  
\mathbb{D}_{\alpha}(\rho \| \sigma)\geqslant  \frac{1-\tilde{\alpha}}{n \tilde{\alpha}}  D_{\tilde{\alpha}}( \mathcal{P}_{\rho^{\otimes n}}\left[\sigma^{\otimes n} \right]\|  \rho^{\otimes n} )\, .
\end{equation}
Taking the limit $n \rightarrow \infty$ we find by Proposition~\ref{prop:char_min_div} that
\begin{equation}  \label{eq:min_2}
\mathbb{D}_{\alpha}(\rho \| \sigma)\geqslant \frac{1-\tilde{\alpha}}{\tilde{\alpha}} \widetilde{D}_{\tilde{\alpha}}(\sigma\| \rho  ) = \widetilde{D'}_{\alpha}(\rho \| \sigma)\, ,
\end{equation}
where we used the symmetry relation~\eqref{eq:sym_rel}  for the last equality. By continuity, we can drop the assumption  $\rho \gg \sigma$ and  inequality~\eqref{eq:min_2} still holds. Since $\mathbb{D}_{\alpha}(\rho \| \sigma)$ and $ \widetilde{D'}_{\alpha}(\rho \| \sigma)$ have the same degree of homogeneity scaling $\rho$ or $\sigma$ respectively, we can also drop the assumption that $\rho$ and $\sigma$ have unit trace.
\end{proof}

\section{Quantum conditional entropies and duality relations} \label{sec:cond_entropies}
Divergences can be used to define conditional entropies, which can be viewed as measures of uncertainty of a system $A$, given the information about a system $B$. Note that we label Hilbert spaces with capital letters $A$, $B$, etc.\ and denote their dimension\footnote{Throughout this theses, we consider finite-dimensional Hilbert spaces only.} by $|A|$, $|B|$, etc.. The set of density operators on $A$, i.e., non-negative operators $\rho_A$ with $\tr \rho_A=1$, is denoted $\cD(A)$.  Then, for any $\rho_{AB} \in \cD(A\otimes B)$ we define the following \emph{quantum conditional R\'enyi entropies} of $A$ given $B$ as
\begin{align} \label{eq:cond_entropies}
&\widebar{H}^{\downarrow}_{\alpha}(A|B)_{\rho}:=-\widebar{D}_{\alpha}(\rho_{AB}\| \id_A \otimes \rho_B) \, , \\*
&\widebar{H}^{\uparrow}_{\alpha}(A|B)_{\rho}:=\sup \limits_{\sigma_B \in \cD(B) }-\widebar{D}_{\alpha}(\rho_{AB}\| \id_A \otimes \sigma_B) \, , \\
&\widetilde{H}^{\downarrow}_{\alpha}(A|B)_{\rho}:=-\widetilde{D}_{\alpha}(\rho_{AB}\| \id_A \otimes \rho_B) \quad \text{and} \quad \\*
&\widetilde{H}^{\uparrow}_{\alpha}(A|B)_{\rho}:=\sup \limits_{\sigma_B \in \cD(B) }-\widetilde{D}_{\alpha}(\rho_{AB}\| \id_A \otimes \sigma_B) \, . 
\end{align}
Note that the special cases $\alpha\in \{0,1,\infty\}$ are defined by taking the limits inside the supremum.\footnote{We are following the notation in~\cite{tomamichel_quantum_2016}. Note that $H_{\text{min}}(A|B)_{\rho| \rho}=\widetilde{H}^{\downarrow}_{\infty}(A|B)_{\rho}$, $H_{\text{min}}(A|B)_{\rho}=\widetilde{H}^{\uparrow}_{\infty}(A|B)_{\rho}$ and $H_{\text{max}}(A|B)_{\rho}=\widetilde{H}^{\uparrow}_{\frac{1}{2}}(A|B)_{\rho}$ are also often used notations.
} 
We call the set of all conditional entropies with $\alpha \in (0,1)$ ``max-like'' and those with $\alpha \in (1,\infty)$ ``min-like'', owing to the fact that under small changes to the state the entropies in either class are approximately equal~\cite{renner_smooth_2004,tomamichel_fully_2009}. 
Moreover, min- and max-like entropies are related by some interesting duality relations, which are summarized in the following lemma.

\begin{mylem} [Duality relations~\cite{tomamichel_fully_2009, tomamichel_relating_2014, muller-lennert_quantum_2013, beigi_sandwiched_2013,konig_operational_2009, berta_single-shot_2008}] \label{lem_duality}

Let $\rho_{ABC}$ be a pure state on $A \otimes B \otimes C$. Then
\begin{align}
\widebar{H}^{\downarrow}_{\alpha}(A|B)_{\rho}+\widebar{H}^{\downarrow}_{\beta}(A|C)_{\rho} =0 &\quad \text{when} \quad \alpha+\beta=2 \, \text{ for } \alpha, \beta \in [0,2] \, \quad \text{and} \nonumber \\
\widetilde{H}^{\uparrow}_{\alpha}(A|B)_{\rho}+\widetilde{H}^{\uparrow}_{\beta}(A|C)_{\rho} =0 &\quad \text{when} \quad \frac{1}{\alpha}+\frac{1}{\beta}=2 \, \text{ for } \alpha, \beta \in [\frac{1}{2},\infty] \, \quad \text{and} \nonumber \\
\widebar{H}^{\uparrow}_{\alpha}(A|B)_{\rho}+\widetilde{H}^{\downarrow}_{\beta}(A|C)_{\rho} =0 &\quad \text{when} \quad \alpha \beta =1 \, \text{ for } \alpha, \beta \in [0,\infty] \, , \nonumber
\end{align} 
where we use the convention that $\frac{1}{\infty}=0$ and $\infty \cdot 0=1\,$.
\end{mylem}

\chapter{Trace inequalities} \label{chap:trace_ineq} 

\section{Introduction}
There are a lot of interesting inequalities between linear operators on Hilbert spaces. For quantum information theory, such operator inequalities that include a trace are especially interesting, since they often give bounds on different information measures. In this chapter, we review different trace inequalities. We consider finite-dimensional Hilbert spaces (and hence matrix inequalities) for simplicity, though most of the results can be extended to separable Hilbert spaces. First, we give a detailed proof of the generalized H\"older inequality for matrices (see e.g.,~\cite[Exercise~IV.2.7]{bhatia_matrix_1997}), and use it to derive a new reversed version of the celebrated Araki-Lieb-Thirring (ALT)  inequality. The reverse ALT inequality then leads to an interesting new relation between the minimal and the Petz  quantum R\'enyi  divergences (which is described in Chapter~\ref{chap:rel_div}). Moreover, we provide a new and elegant proof of the logarithmic form of the reverse Golden-Thompson inequality. Let us first introduce some notation.

\section{Schatten norms} \label{sec:Schatten}
The \emph{Schatten} $p$\emph{-norm} of any matrix $M \in \Mat(n,n)$ is given by 
\begin{align} \label{eq:defi_schatten_norm}
\norm{M}_p:=\big(\tr |M|^p\big)^{\frac{1}{p}} \quad \text{for} \quad p\geqslant 1 \ ,
\end{align}
where $|M|:=\sqrt{M^{*} M}$. 
We may extend this definition to all $p>0$, but note that $\norm{M}_p$ is not a norm for $p \in (0,1)$ since it does not satisfy the triangle inequality.
In the limit $p\to \infty$ we recover the \emph{operator norm} and for $p=1$ we obtain the \emph{trace norm}. 
Schatten norms are functions of the singular values and thus unitarily invariant. Moreover, they satisfy $\|M\|_p = \|M^{*}\|_p$ and $\|M\|_{2p}^2 = \|MM^{*}\|_p= \|M^{*} M\|_p\, $.

\section{H\"older inequality for Schatten norms}
In the following, we fill in the details of the proof of the generalized H\"older inequality, which was stated as Exercise~IV.2.7 in~\cite{bhatia_matrix_1997}.
\begin{mydef} \label{def:unitarily_inv_norm}
A norm $\normT{ \cdot}$ on $\Mat(n,n)$ is called \emph{unitarily invariant}  if $\normT{ UAV }=\normT{A }$ for all  $A \in \textnormal{Mat}(n,n)$ and $U,V \in \textnormal{U}(n)\,$.
\end{mydef}
\begin{mythm} \label{thm:Hoelder_UI_norm}
Let $\normT{ \cdot}$ be a \emph{unitarily invariant} norm on Mat$(n,n)$. Let $s$, $s_1,\dots,s_l$ be positive real numbers and $\{A_k\}_{k=1}^l$ be a collection of $n\times n$ matrices. Then
\begin{align} 
\normT{ \, | \prod_{k=1}^l A_k|^s \,}^{\frac{1}{s}} \leqslant \prod_{k=1}^l \normT{  |A_k|^{s_k} }^{\frac{1}{s_k}} \, ,\qquad \text{for}\quad \sum_{k=1}^l\frac1{s_k}=\frac1s \, .
\end{align}
\end{mythm}
Setting $\normT{\cdot}=\norm{\cdot}_1$ in Theorem~\ref{thm:Hoelder_UI_norm}, we recover the generalized H\"older inequality for Schatten (quasi)-norms.
\begin{mycor} [Generalized  H\"older inequality] \label{cor:generalized_Holder}
Let $s$, $s_1,\dots,s_l$ be positive real numbers (where we also allow $\infty$ using the convention that $\frac{1}{\infty}=0$) and $\{A_k\}_{k=1}^l$ be a collection of $n\times n$ matrices. Then
\begin{align}
\label{eq:hoelder}
\norm{ \prod_{k=1}^l A_k}_s \leqslant \prod_{k=1}^l \norm{ A_k }_{s_k} \, ,\qquad \text{for}\quad \sum_{k=1}^l\frac1{s_k}=\frac1s \, .
\end{align}
\end{mycor}
Note that the cases where some parameters are equal to $\infty$ follow as limit cases of the finite ones.

To prove Theorem~\ref{thm:Hoelder_UI_norm}, we need some preparative results.
\begin{mydef} [Symmetric gauge function~\cite{bhatia_matrix_1997}] \label{def:sym_gauge_fct}
A function $\Phi: \mathbb{R}^n \rightarrow \mathbb{R}_{+}$ is called a \emph{symmetric gauge function} if
\begin{enumerate}[(i)]
\item $\Phi$ is a norm on the real vector space $\mathbb{R}^n\,$,
\item $\Phi(Px)=\Phi(x)$ for all $x \in \mathbb{R}^n, \,  P\in S_n\,$ and
\item $\Phi(\epsilon_1x_1,\dots, \epsilon_n x_n)=\Phi(x_1,\dots,x_n)$ if $\epsilon_j=\pm1 \,$,
\end{enumerate}
where $S_n$ denotes the group of all $n \times n$ permutation matrices. In addition, we will always assume that $\Phi$ is normalized
\begin{enumerate}[(i)]
 \setcounter{enumi}{3}
\item $\Phi(1,0,\dots,0)=1$.
\end{enumerate}
\end{mydef}
\begin{myprop} [Problem II.5.11 (iv)~\cite{bhatia_matrix_1997}] \label{prop:gauge_fct_monoton}
Every symmetric gauge function is monoton on $\mathbb{R}_{+}^n$, i.e., for $x,y \in \mathbb{R}_{+}^n$ with $x\leqslant y$,\footnote{Note that $x\leqslant y$ is to be understood on a per-element basis, i.e., $x_i\leqslant y_i$ for all $i=1,\dots,n \,$. } we have that $\Phi(x)\leqslant \Phi(y)$.
\end{myprop}
\begin{proof}
By the property i) and iii) of Definition~\ref{def:sym_gauge_fct}, $\Phi$ is a norm that satisfies $\Phi(|x|)=\Phi(x)$.\footnote{Note that the absolute value in the expression $|x|$ is to be understood as taking the absolute value of each element of $x\,$.} The Proposition now follows directly from the real case of Proposition~IV.1.1 of~\cite{bhatia_matrix_1997}.
\end{proof}
\begin{mydef} [ $x \prec_w y$]
Let $x=(x_1,x_2,\dots,x_n)\in \mathbb{R}^n$ and let $x^{\downarrow}$ denote the vector that one gets by reordering the entries of $x$ in a decreasing order. We say that $x$ is \emph{weakly submajorised} by a vector $y \in \mathbb{R}^n$ (written as $x \prec_w y$), if for all $1\leqslant k\leqslant n \,$, we have that
\begin{align}
\sum_{i=1}^k x_i^{\downarrow}\leqslant \sum_{i=1}^k y_i^{\downarrow} \, .
\end{align}
\end{mydef}
Note that every symmetric gauge function is convex (because it satisfies the triangle inequality). Then, Proposition~\ref{prop:gauge_fct_monoton} together with Theorem~II.3.3 of~\cite{bhatia_matrix_1997} imply that a symmetric gauge function $\Phi$ is \emph{strongly isotone}, i.e., we have that
\begin{align} \label{eq:isotone}
x \prec_w y \implies \Phi(x)\leqslant \Phi(y) \, .
\end{align}
\begin{mythm} [Exercise IV.1.7~\cite{bhatia_matrix_1997}] \label{thm:gauge_fct}
Let $p,q,r$ be positive real numbers with $\frac{1}{p}+\frac{1}{q}=\frac{1}{r}$. Let $x,y \in \mathbb{R}^n$. Then, for every symmetric gauge function $\Phi$, we have
\begin{align}
\Phi(|x \cdot y|^r)^{\frac{1}{r}}\leqslant \Phi(|x|^p)^{\frac{1}{p}} \Phi(|x|^q)^{\frac{1}{q}}.
\end{align}
\end{mythm}
The proof of Theorem~\ref{thm:gauge_fct} works similar to the proof of Theorem~IV.1.6 of~\cite{bhatia_matrix_1997}. First, note that for a convex function $f:I \rightarrow \mathbb{R}$ on an interval $I\subset \mathbb{R}$ and positive real numbers $a_1,\dots,a_n$ with $\sum_{i=1}^n a_i=1$ we have that
\begin{align}
f(\sum_{i=1}^n a_i t_i)\leqslant \sum_{i=1}^n a_i f(t_i) \quad \textnormal{for all }t_i\in I\, .
\end{align}
Setting $f(t)=-\log(t)$ and $I=(0,\infty)$, we find
\begin{align} \label{eq:ag_mean_inequality}
\prod_{i=1}^n t_i^{a_i} \leqslant \sum_{i=1}^n a_i t_i \quad \textnormal{for }t_i\geqslant0\, ,
\end{align}
which is called the (weighted) \emph{arithmetic-geometric mean inequality}.
\begin{proof} [Proof of Theorem~\ref{thm:gauge_fct}]
Let $p,q,r$ be positive real numbers with $\frac{1}{p}+\frac{1}{q}=\frac{1}{r}$. Let $x,y \in \mathbb{R}^n$ and let $\Phi$ be a symmetric gauge function. Setting $n=2$ and $a_1=\frac{r}{p}, a_2=\frac{r}{q}$ in~\eqref{eq:ag_mean_inequality}, we find
\begin{align} 
|x\cdot y|^r=|x^p|^{\frac{r}{p}}|y^q|^{\frac{r}{q}}\leqslant  \frac{r}{p} |x|^p +\frac{r}{q}  |y|^q \, ,
\end{align}
where we take the multiplication, the norm and the powers again element by element. Hence, using Proposition~\ref{prop:gauge_fct_monoton} (and that a symmetric gauge function is a norm), we find that
\begin{align} \label{eq:Phi_inequality}
\Phi(|x \cdot y|^r)\leqslant  \frac{r}{p} \Phi (|x|^p) +\frac{r}{q} \Phi( |y|^q) \, .
\end{align}
Let $t>0$ and note that the left hand side of~\eqref{eq:Phi_inequality} is invariant under the substitution $x \rightarrow tx$ and $y \rightarrow t^{-1}y$. Therefore
\begin{align}
\Phi(|x \cdot y|^r)\leqslant \inf_{t>0} \left[ \frac{r t^p }{p}  \Phi (|x|^p)  +\frac{r }{q t^q} \Phi( |y|^q)  \right] \, .
\end{align}
Searching for a local minimum by differentiation suggests to set $t=\left( \frac{\Phi( |y|^q) }{ \Phi (|x|^p) } \right)^{\frac{1}{p+q}}$. This yields 
\begin{align}
\Phi(|x \cdot y|^r)\leqslant \Phi (|x|^p)^{\frac{r}{p}}  \Phi( |y|^q)^{\frac{r}{q}},
\end{align}
which prooves the claim. \end{proof}
\begin{mythm}[Theorem~IV.2.5 of~\cite{bhatia_matrix_1997}] \label{thm:submojorised}
Let $A,B \in \Mat(n,n)$. Then
\begin{align}
\s^r(AB) \prec_w  \s^{r}(A)\s^{r}(B) \qquad \textnormal{for all } r>0 \, ,
\end{align}
where s$(C)$ denotes the vector whose entries correspond to the singular values of a matrix $C \in \Mat(n,n)$.
\end{mythm}
\begin{mylem}[Exercise~IV.2.7 of~\cite{bhatia_matrix_1997}] \label{thm:Exc_Hoelder}
Let $p,q,r$ be positive real numbers with $\frac{1}{p}+\frac{1}{q}=\frac{1}{r}$ and $\normT{\cdot}$ be a unitarily invariant norm. Then
\begin{align}
\normT{|AB|^r}^{\frac{1}{r}}\leqslant \normT{|A|^p}^{\frac{1}{p}} \normT{|B|^q}^{\frac{1}{q}}.
\end{align} 
\end{mylem}
\begin{proof}
Let $p,q,r$ be positive real numbers with $\frac{1}{p}+\frac{1}{q}=\frac{1}{r}$ and $\normT{\cdot}$ be a unitarily invariant norm. From Theorem~\ref{thm:submojorised}, we have 
\begin{align}
\s^r(AB) \prec_w  \s^{r}(A)\s^{r}(B) \, .
\end{align} 
Let us define the function $\Phi(x)=\normT{\diag(x)}:\mathbb{R}^n\rightarrow \mathbb{R}_{+}$, where $\diag(x)$ denotes a diagonal matrix with diagonal entries $x_1,x_2,\dots,x_n$. By Theorem~IV.2.1 of~\cite{bhatia_matrix_1997}, $\Phi$ is a symmetric gauge function. Therefore, $\Phi$ is strongly isotone (cf.~\eqref{eq:isotone}), which leads to 
\begin{align} \label{eq:Phi_from_isotone}
\Phi(\s^r(AB))^{\frac{1}{r}} \leqslant  \Phi \left(\s^r(A) \s^r(B)\right)^{\frac{1}{r}} \,.
\end{align} 
Using Theorem~\ref{thm:gauge_fct}, we can bound the right hand side of \eqref{eq:Phi_from_isotone} by
\begin{align} \label{bound_rhs_Phi}
\Phi(\s^r(A)\s^r(B)^r)^{\frac{1}{r}}  \leqslant  \Phi\left(\s(A)^p \right)^{\frac{1}{p}} \Phi\left(\s(B)^q\right)^{\frac{1}{q}}.
\end{align} 
Rcall that for $C \in \Mat(n,n)$ and a positive real number $t$, we have that $\s^t(C)=\s^{\frac{t}{2}}(CC^{*})$. Since the singular values of a non negative matrix are equal to its eigenvalues, we find that  $\s^t(C)= \s(|C|^t)$. Applying this to the left hand side of~\eqref{eq:Phi_from_isotone} and the right hand side of~\eqref{bound_rhs_Phi}, we find
\begin{align} 
\Phi(\s(|AB|^r))^{\frac{1}{r}}  \leqslant  \Phi\left(\s(|A|^p) \right)^{\frac{1}{p}} \Phi\left(\s(|B|^q)\right)^{\frac{1}{q}}.
\end{align} 
The claim now follows by noting that $\Phi(\s(C))=\normT{\diag\left(\s(C)\right)}=\normT{C}$  for any $C \in \Mat(n,n)$, since $\normT{\cdot}$ is unitarily invarint.
\end{proof}
Theorem~\ref{thm:Hoelder_UI_norm} now follows directly from Lemma~\ref{thm:Exc_Hoelder} by induction.

\section{Reverse Araki-Lieb-Thirring inequality} \label{sec:reverse_ALT}

Let us first recall the statement of the celebrated  Araki-Lieb-Thirring (ALT) inequality. 

\begin{mythm} [ALT inequality~\cite{lieb_inequalities_1976,araki_inequality_1990}\footnote{See also~\cite{SBT15} for an intuitive proof of the ALT inequality and a multivariate form of it.}] \label{thm:ALT}
Let  $A$ and $B$ be  positive semi-definite matrices (of the same dimension) and  $q\geqslant 0$. Then
\begin{align}
 \tr\, (B^{\frac{r}{2}} A^r B^{\frac{r}{2}})^q \leqslant \tr \, (B^{\frac{1}{2}} A B^{\frac{1}{2}})^{rq}  & \quad \text{ for $r \in [0,1] $ and}  \label{eq:alt_smaller_one} \\
\tr\, (B^{\frac{r}{2}} A^r B^{\frac{r}{2}})^q \geqslant \tr \, (B^{\frac{1}{2}} A B^{\frac{1}{2}})^{rq} &\quad \text{ for $r \geqslant 1 \, .$ } \label{eq:alt_bigger_one}
\end{align}
\end{mythm}
Our main result of this chapter is a reversed version of the ALT inequality, which follows from the H\"older inequality.

\begin{mythm}[Reverse ALT inequality]  \label{thm:Inv_ALT}
Let $A$ and $B$ be positive semi-definite matrices and $q > 0$. Then, for $r \in (0,1]$ and $a,b\in (0,\infty]$ such that $\frac{1}{2rq}=\frac{1}{2q}+\frac{1}{a}+\frac{1}{b}$, we have
\begin{align}
\label{eq:Inv_ALT}
\tr\,\big(B^{\frac{1}{2}}A B^{\frac{1}{2}}\big)^{rq} \leqslant \Big(\tr\,\big(B^{\frac{r}{2}}A^r B^{\frac{r}{2}}\big)^q\Big)^r \norm{A^{\frac{1-r}{2}}}_a^{2rq}\norm{B^{\frac{1-r}{2}} }_b^{2rq} \, .
\end{align}
Meanwhile, for $r \in [1,\infty)$ and  $a,b\in (0,\infty]$ such that $\frac{1}{2q}=\frac{1}{2rq}+\frac{1}{a}+\frac{1}{b}$, we have 
\begin{align}
\label{eq:Inv_ALT_r_bigger_one}
\tr\,\big(B^{\frac{1}{2}}A B^{\frac{1}{2}}\big)^{rq} \geqslant \Big(\tr\,\big(B^{\frac{r}{2}}A^r B^{\frac{r}{2}}\big)^q\Big)^r \norm{A^{\frac{r-1}{2}}}_a^{-2rq}\norm{B^{\frac{r-1}{2}} }_b^{-2rq} \, .
\end{align}
\end{mythm}

\begin{proof}
For $r=1$ the statement is trivial.
Let  $r \in (0,1)$ and $q > 0$. 
 We can rewrite the trace-terms in~\eqref{eq:Inv_ALT} as Schatten (quasi-)norms  
\begin{align}
&\tr\,\big(B^{\frac{1}{2}}A B^{\frac{1}{2}}\big)^{rq}=\norm{ B^{\frac{1}{2}}  A^{\frac{1}{2}}}_{2rq}^{2rq} \qquad \text{and}\\*
&\tr\,\big(B^{\frac{r}{2}}A^r B^{\frac{r}{2}}\big)^{q}=\norm{ B^{\frac{r}{2}}  A^{\frac{r}{2}}}_{2q}^{2q} \, .
\end{align}
Inequality~\eqref{eq:Inv_ALT} then follows by an application of the generalized H\"older inequality given in Corollary~\ref{cor:generalized_Holder} with $l=3$. Choosing $s=2rq$, and $s_1=b$, $s_2=2q$, and $s_3=a$ for some $a,b \in (0,\infty]$ with $\frac{1}{2rq}=\frac{1}{2q}+\frac{1}{a}+\frac{1}{b}$, we find
\begin{align}
\tr\,\big(B^{\frac{1}{2}}A B^{\frac{1}{2}}\big)^{rq}=\norm{  B^{\frac{1-r}{2}}  B^{\frac{r}{2}}  A^{\frac{r}{2}} A^{\frac{1-r}{2}} }_{2rq}^{2rq} \leqslant \norm{B^{\frac{1-r}{2}} }_b^{2rq} \norm{B^{\frac{r}{2}}  A^{\frac{r}{2}}}_{2q}^{2rq} \norm{A^{\frac{1-r}{2}} }_a^{2rq} \, .\nonumber
\end{align}
Inequality~\eqref{eq:Inv_ALT_r_bigger_one} now follows  from~\eqref{eq:Inv_ALT} by substituting  $A \rightarrow A^r$, $B \rightarrow B^r$, $r \rightarrow \frac{1}{r}\,$ and $q\rightarrow qr$. 
 \end{proof}

 \begin{myrmk} \label{rmk:Inv_ALT}
Another reverse ALT inequality was given in~\cite{audenaert_araki-lieb-thirring_2008}, where it was shown that for $r \in (0,1)$ and $q > 0$ we have
\begin{align}
\label{eq:adenauer}
\tr(B^{\frac{1}{2}}A B^{\frac{1}{2}})^{rq} \leqslant \big(\tr(B^{\frac{r}{2}}A^r B^{\frac{r}{2}})^q\big)^r \big(\tr \, A^{rq}\norm{B}_{\infty}^{rq}\big)^{1-r} \, ,
\end{align}
while for $r>1$ the inequality holds in the opposite direction. We recover these inequalities as a corollary of Theorem~\ref{thm:Inv_ALT} by setting $b=\infty$ and $a=\frac{2rq}{1-r}$ in~\eqref{eq:Inv_ALT}, and $b=\infty$ and $a=\frac{2rq}{r-1}$ in~\eqref{eq:Inv_ALT_r_bigger_one}. We note that there also exists a reverse ALT inequality in terms of matrix means (see e.g.~\cite{ando94}) that however is different to Theorem~\ref{thm:Inv_ALT}. 
\end{myrmk}

\section{Reverse Golden-Thompson inequality}

Let us first recall the celebrated Golden-Thompson (GT) inequality.

\begin{mythm} [GT inequality~\cite{golden65,thompson65}] \label{thm:GT}
Let  $H_1$ and $H_2$ be two Hermitian matrices. Then
\begin{align} \label{eq_GT}
\tr\, \ee^{H_1 + H_2} \leqslant \tr \, \ee^{H_1} \ee^{H_2} \, .
\end{align}
\end{mythm}
Note that we have equality in~\eqref{eq_GT} if and only if $H_1$ and $H_2$ commute. Interestingly, the GT inequality follows directly from the ALT inequality together with the Lie-Trotter product formula (cf.~\cite{reed_functional_1980}, page 295). On the other hand, there is a reverse GT inequality, which is related to matrix means.
\begin{mythm} [Reverse GT inequality~\cite{HP93}] \label{thm:reverse_GT}
Let   $p>0$ and $0\leqslant \alpha \leqslant 1$. Let $H_1$ and $H_2$ be two Hermitian matrices. Then
\begin{align} \label{eq_reverse_GT}
\tr \, \left( \ee^{ p H_1} \#_{\alpha} \ee^{p H_2} \right)^{\frac{1}{p}} \leqslant \tr\, \ee^{(1-\alpha)H_1 + \alpha H_2} \, .
\end{align}
\end{mythm}
To proof the reverse GT inequality, it is translated into a logarithmic trace inequality in~\cite{HP93}. We give a new and simple proof of this logarithmic trace inequality in the next section. 

\subsection{Logarithmic form of the reverse GT inequality}
In~\cite{HP93}, it is shown that the reverse GT inequality stated in Theorem~\ref{thm:reverse_GT} for a fixed $p$ is equivalent to the following logarithmic trace inequality for the same $p$.

\begin{mythm} [Logarithmic form of reverse GT~\cite{HP93}] \label{thm:rev_GT_log}
Let  $A$ and $B$ be positive definite matrices.\footnote{For an extension of~\eqref{eq_reverse_GT_log} to positive semi-definite matrices see Section~4 of~\cite{HP93}. } Then
\begin{align} \label{eq_reverse_GT_log}
 \frac{1}{p} \tr \, A \log A^{\frac{p}{2}} B^p  A^{\frac{p}{2}} \geqslant \tr \, A \left(\log A + \log B \right) \, . 
\end{align}
\end{mythm}
\begin{proof}[Alternative proof to~\cite{HP93} of Theorem~\ref{thm:rev_GT_log} (for $p=1$).]
First, recall from Proposition~\ref{prop:limit_min} that we have 
\begin{align}
\frac{1}{\tr\, A} \tr\, A \left( \log A -  \log B \right) &= \lim_{\alpha \downarrow 1} \tilde D_{\alpha}(A \| B) \nonumber \\
&= \lim_{\alpha \downarrow 1} \frac{1}{\alpha -1} \log \frac{1}{\tr\, A} \tr\left(B^{\frac{1-\alpha}{2 \alpha}} A B^{\frac{1-\alpha}{2 \alpha}} \right)^{\alpha} \, , \nonumber 
\end{align}
and from Proposition~\ref{prop:limit_max} that we have
\begin{align}
\frac{1}{\tr\, A}  \tr \, A \log A^{\frac{1}{2}} B^{-1}  A^{\frac{1}{2}}&=  \lim_{\alpha \downarrow 1} \widehat D_{\alpha}(A \| B) \nonumber \nonumber \\
&= \lim_{\alpha \downarrow 1} \frac{1}{\alpha -1} \log \frac{1}{\tr\, A} \tr \, B^{\frac{1}{2}}\left(B^{-\frac{1}{2}} A B^{-\frac{1}{2}} \right)^{\alpha} B^{\frac{1}{2}} \, . \nonumber 
\end{align}
Let  us set $\tilde{B}=B^{-1}$. Then we find
\begin{align}
&\frac{1}{\tr\, A} \left( \tr\, A \log A + \tr \, A \log B \right)\\*
&\hspace{8em}= \lim_{\alpha \downarrow 1} \frac{1}{\alpha -1} \log \frac{1}{\tr\, A} \tr\left(\tilde{B}^{\frac{1-\alpha}{2 \alpha}} A \tilde{B}^{\frac{1-\alpha}{2 \alpha}} \right)^{\alpha} \\*
&\hspace{8em}= \lim_{\alpha \downarrow 1} \frac{1}{\alpha -1} \log \frac{1}{\tr\, A} \tr\left(\tilde{B}^{\frac{1}{2 \alpha}}  \tilde{B}^{-\frac{1}{2}}  A \tilde{B}^{-\frac{1}{2}} \tilde{B}^{\frac{1}{2 \alpha}} \right)^{\alpha} \\*
&\hspace{8em} \leq \lim_{\alpha \downarrow 1} \frac{1}{\alpha -1} \log  \frac{1}{\tr\, A} \tr \, \tilde{B}^{\frac{1}{2}}\left(\tilde{B}^{-\frac{1}{2}} A \tilde{B}^{-\frac{1}{2}} \right)^{\alpha} \tilde{B}^{\frac{1}{2}}  \label{eq:ALT_rev_GT}\\*
&\hspace{8em}=\frac{1}{\tr\, A} \tr A \log A^{\frac{1}{2}} \tilde{B}^{-1} A^{\frac{1}{2}}\\*
&\hspace{8em}=\frac{1}{\tr\, A}  \tr A \log A^{\frac{1}{2}} B A^{\frac{1}{2}} \, ,
\end{align}
where we applied the ALT inequality~\eqref{eq:alt_bigger_one} with $q=1$, $r=\alpha$, $A':=\tilde{B}^{-\frac{1}{2}} A \tilde{B}^{-\frac{1}{2}}$ and $B':=\tilde{B}^{\frac{1}{ \alpha}}$. This proves Theorem~\ref{thm:rev_GT_log} for $p=1$.
\end{proof}
 We conclude that we found a new and elegant proof  (based on the ALT inequality) of the logarithmic form of the reverse GT inequality for $p=1$.
\begin{myrmk}
The inequality~\eqref{eq:ALT_rev_GT} shows that  $\widetilde{D}_{\alpha}(\rho \| \sigma) \leqslant \widehat{D}_{\alpha} (\rho \| \sigma)$ for $\alpha > 1$ and for positive operators $\rho$ and $\sigma$. This relation is well known (see for example~\cite{tomamichel_quantum_2016}), but the proof technique is new.
\end{myrmk}

\subsection{Open question: multivariate reverse GT inequality}
Recently, an elegant generalization of the ALT inequality for an arbitrary number of matrices was found by using pinching techniques (or results from complex analysis based on the maximum principle for holomorphic functions)~\cite{SBT15}. As in the two matrix case, the multivariate ALT inequality leads to an multivariate GT inequality by using a multivariate Lie-Trotter product formula, which leads to an interesting lower bound on the conditional quantum mutual information~\cite{SBT15}. Using similar techniques as in~\cite{SBT15}, a multivariate reverse GT inequality would lead to an upper bound on the conditional quantum mutual information. Alternatively, a multivariate generalization of the logarithmic form of the reverse GT inequality would lead even more directly to such a bound. \\

In 2009, a multivariate reverse GT inequality was derived in~\cite{HIAI20093105} and generalized in~\cite{hiai_log-majorization_2016} in 2016. It is based on a natural generalization of the geometric mean, the well studied \emph{Karcher mean}~\cite{lawson_monotonic_2011,Yamazaki2013}. The Karcher mean of positive definite matrices $A_1,A_2,\dots,A_n$ with a weight $(\omega_1,\dots,\omega_n)$ is defined as

\begin{align} 
\text{G}_{\omega} \left( A_1, A_2, \dots,  A_n \right):=\underset{Z>0}{\text{arg min}} \sum_{j=1}^n \omega_j \text{d}^2(Z,A_j) \, ,
\end{align}

where the Riemannian trace metric d$(A,B)$ on the set of positive definite matrices is given by

\begin{align} 
\text{d}(A,B):= \norm{\log A^{-\frac{1}{2}} B A^{-\frac{1}{2}}}_2 \, .
\end{align}

Then, the following theorem follows as a special case of the results derived in~\cite{HIAI20093105,hiai_log-majorization_2016}.

\begin{mythm} [Multivariate reverse GT inequality~\cite{HIAI20093105,hiai_log-majorization_2016}] \label{thm:multivariate_reverse_GT}
Let   $p>0$ and \linebreak $H_1,H_2, \dots H_n$ be Hermitian matrices. Then
\begin{align} \label{eq:multivariate_reverse_GT}
\tr \, \text{G}_{\omega} \left( \ee^{ p H_1}, \ee^{p H_2}, \dots,  \ee^{p H_n} \right)^{\frac{1}{p}} \leqslant \tr\, \ee^{ \sum_{j=1}^n \omega_j H_j } \, .
\end{align}
\end{mythm}

Unfortunately, there is no explicit formula known for the Karcher mean for more than two matrices, which seems to restrict the usefulness of inequality~\eqref{eq:multivariate_reverse_GT} (for more than two matrices) for applications in quantum information theory. \\

It would be very interesting to find an explicit multivariate version of the reverse GT inequality or of its logarithmic form. Our simple proof of~\eqref{eq_reverse_GT_log} looks promising to achieve this goal, since it is based on the ALT inequality, for which a multivariate version is known~\cite{SBT15}. Moreover, it might be possible to use similar pinching techniques as used in~\cite{SBT15} to generalize~\eqref{eq:ALT_rev_GT} to more than two matrices.

\chapter{Relations between quantum R\'enyi divergences} \label{chap:rel_div}

\section{Introduction} \label{sec: intro}
A natural and important question is how the different families of quantum R\'enyi divergences introduced in Section~\ref{sec:important_divergences} are related to each other. As we will see in this chapter, this question is strongly related to mathematical trace inequalities.

The ALT inequality~\cite{lieb_inequalities_1976,araki_inequality_1990} implies that the Petz divergence is larger than or equal to the minimal divergence, i.e., $\widebar{D}_{\alpha}(\rho \| \sigma) \geqslant \widetilde{D}_{\alpha}(\rho \| \sigma)$. 
But what remains unanswered is how much bigger than the minimal divergence the Petz divergence can be. 
In Section~\ref{sec:Petz_Min}, we settle this question for $\alpha \leq1$ by showing that $\widebar{D}_{\alpha}(\rho \| \sigma) \leqslant \frac{1}{\alpha} \widetilde{D}_{\alpha}(\rho \| \sigma)$ if $\rho$ and $\sigma$ are normalized. (Note that we shall make use of the convention $\frac{1}{0}=\infty$ in this chapter.)
This result follows from the reverse ALT inequality stated in Theorem~\ref{thm:Inv_ALT}. The reverse bound between the minimal and the Petz quantum R\'enyi divergence  leads then to new relations between quantum conditional R\'enyi entropies, which are discussed in Section~\ref{sec:rel_entropies}.

Moreover, the new bound between the minimal and the Petz divergence leads to a unified picture of the relationship between pretty good quantities and their optimal versions in quantum information theory. We will discuss this implication in detail in Chapter~\ref{cha:pgm}.

\vspace{3mm}

\section{Relation between Petz and minimal quantum R\'enyi divergence} \label{sec:Petz_Min}
We have seen in Section~\ref{sec:min_max_divergence} that the minimal  quantum R\'enyi divergence provides a lower bound for all other quantum R\'enyi divergences satisfying the DPI. Hence, in particular, we have $\widetilde{D}_{\alpha}(\rho \| \sigma)\leqslant \widebar{D}_{\alpha}(\rho \| \sigma)$ for all $\alpha \in [0,\infty]$ and non-negative operators $\rho$ and $\sigma$.\footnote{Alternatively, this follows directly from the ALT inequality.} Theorem~\ref{thm:Inv_ALT} leads to reversed relations between these two divergences. In the case where $\alpha \in [0,1]$, we find a particularly  useful relation of a simple form.  
\begin{mycor}  \label{cor:Upper_bound}
Let $\rho \neq 0$ and $\sigma$ be two non-negative operators and $\alpha \in [0,1]$. Then
\begin{align}
\label{eq:Upper_bound}
\alpha \widebar{D}_{\alpha}(\rho\|\sigma)+(1-\alpha) (\log \tr \rho - \log \tr \sigma) \leqslant \widetilde{D}_{\alpha}(\rho\|\sigma) \leqslant \widebar {D}_{\alpha}(\rho\|\sigma) \, .
\end{align}
\end{mycor}
\begin{proof}
The second inequality is a direct consequence of the ALT inequality. 
It thus remains to show the first inequality. We note that it suffices to consider the case $\alpha \in (0,1)$, as $\alpha \in \{0,1\}$ then follows by continuity. By definition, we can reformulate the first inequality of~\eqref{eq:Upper_bound} as 
\begin{align} \label{eq:reformulated_thm}
\widetilde{Q}_{\alpha}(\rho\|\sigma) \leqslant \widebar{Q}_{\alpha}(\rho\|\sigma)^{\alpha}  (\tr \rho )^{\alpha(1-\alpha)}  (\tr \sigma )^{(1-\alpha)^2} \, . 
\end{align}
This follows from Theorem~\ref{thm:Inv_ALT} with $q=1$, $r=\alpha$, $A=\rho$,  $B=\sigma^{\frac{1-\alpha}{\alpha}}$, $a=\frac{2}{1-\alpha}$, and $b=\frac{2\alpha}{(1-\alpha)^2}$.
 \end{proof}

There is a well known equality condition for the ALT inequality, which leads to an equality condition for the second inequality of~\eqref{eq:Upper_bound}.
 \begin{mylem} \label{lem:eq_cond_ALT}
For $\alpha \in (0,1)$, we have  $\widetilde{D}_{\alpha}(\rho\|\sigma) = \widebar {D}_{\alpha}(\rho\|\sigma)$ if and only if $\rho$ and $\sigma$ commute. 
\end{mylem}
\begin{proof}
To see this, note that for  $r \in (1,\infty)$ and $rq \geqslant 1$, we have equality in the ALT inequality~\eqref{eq:alt_bigger_one} if and only if $A$ and $B$ commute. Equality for commuting states is obvious; for the other direction, note that we can rewrite~\eqref{eq:alt_bigger_one} using the substitution $rq = q'$ as
 \begin{align}
 \label{eq:alt_norm_version}
\norm{ (B^{\frac{r}{2}} A^r B^{\frac{r}{2}})^{\frac{1}{r}}}_{q'}\ \geqslant \norm{ (B^{\frac{1}{2}} A B^{\frac{1}{2}})}_{q'}\, .
\end{align}
 Equality in the inequality~\eqref{eq:alt_norm_version} for some  $r \in (1,\infty)$ (and noting that we have also equality for $r=1$) implies that the function $r \mapsto \| (B^{\frac{r}{2}} A^r B^{\frac{r}{2}})^{\frac{1}{r}}\|_{q'} $ is not strictly increasing.  Therefore, by \cite[Theorem 2.1]{hiai_equality_1994}, it follows\footnote{Here we use our assumption that  $q' \geqslant 1$, since in this case $\norm{\cdot}_{q'}$ is a strictly increasing norm.} that $[A,B]=0$. Let  $\rho,\sigma$ be non negative. Setting $r=\nicefrac{1}{\alpha},q=\alpha$ and $A=\rho^{\alpha}$, $B=\sigma^{1-\alpha}$ in~\eqref{eq:alt_bigger_one}, we conclude that for $\alpha \in (0,1)$ we have that $\widetilde{D}_{\alpha}(\rho\|\sigma) = \widebar {D}_{\alpha}(\rho\|\sigma)$ if and only if $[\rho,\sigma]=0$.
 \end{proof}

For density operators $\rho$ and $\sigma$ the first inequality of Corollary~\ref{cor:Upper_bound} simplifies to 
\begin{align}
\alpha \widebar{D}_{\alpha}(\rho \| \sigma) \leqslant \widetilde{D}_{\alpha}(\rho \| \sigma) \quad \text{for} \quad \alpha \in [0,1]\, .
\end{align}
This bound is simpler than an alternative bound given in~\cite{mosonyi_coding_2015}, which is based on the earlier reversed ALT inequality in \eqref{eq:adenauer} and states that 
$ \alpha \widebar{D}_{\alpha}(\rho \| \sigma) -\log \tr \rho^{\alpha} +(\alpha-1) \log \norm{\sigma}_{\infty} \leqslant \widetilde{D}_{\alpha}(\rho \| \sigma)$ for density operators $\rho$ and $\sigma$.

\section{Relations between different quantum entropies} \label{sec:rel_entropies}

The new bound between the minimal and the Petz divergence given in Corollary~\ref{cor:Upper_bound} leads to interesting new relations between max-like entropies, which are described in Section~\ref{sec:max_like}. In the following, we use the notation introduced in Section~\ref {sec:cond_entropies}.\footnote{In particular $A$ and $B$ refer to quantum systems (Hilbert spaces) in the following, and do not denote matrices (as this was the case in Chapter~\ref{chap:trace_ineq}).} In particular, recall that we call the set of all conditional entropies with $\alpha \in (0,1)$ ``max-like'' and those with $\alpha \in (1,\infty)$ ``min-like''. By duality relations (cf. Lemma~\ref{lem_duality}), the relations between the max-like entropies lead to new bounds for min-like entropies (cf. Section~\ref{sec:min_like}). In addition, we introduce an equality condition for quantum R\'enyi entropies in Section~\ref{sec:equality_cond_entropy}, which will lead to an  equality condition for pretty good measures of a simpel form (cf. Section~\ref{sec:opt_cond_pgm}).

\subsection{Relations between max-like entropies} \label{sec:max_like}
As a direct consequence of Corollary~\ref{cor:Upper_bound}, we find the following relation between conditional max-like entropies.

\begin{mycor} \label{cor:entropies}
For $\alpha \in [0,1]$ and $\rho_{AB}\in \cD(A\otimes B)\,$, we have that
\begin{align}
&\widebar{H}^{\downarrow}_{\alpha}(A|B)_{\rho} \leqslant \widetilde{H}^{\downarrow}_{\alpha}(A|B)_{\rho}\leqslant \alpha \widebar{H}^{\downarrow}_{\alpha}(A|B)_{\rho}+(1-\alpha)\log |A|  \label{eq:Bounds_entropies1} \qquad \text{and}  \\
&\widebar{H}^{\uparrow}_{\alpha} (A|B)_{\rho}\leqslant \widetilde{H}^{\uparrow}_{\alpha}(A|B)_{\rho}\leqslant \alpha \widebar{H}^{\uparrow}_{\alpha}(A|B)_{\rho}+(1-\alpha)\log |A| \label{eq:Bounds_entropies2} \, .
\end{align}
\end{mycor}

We can further improve the upper bounds in~\eqref{eq:Bounds_entropies1} and~\eqref{eq:Bounds_entropies2} by removing the second term if $\rho_{AB}$ has a special structure consisting of a quantum and a  classical part that is handled coherently.

\begin{myprop} \label{ccq_states}
Let $\ket{\rho}_{XX'BB'}=\sum_x \sqrt{p_x}\ket{x}_X\ket{x}_{X'}\ket{\xi_x}_{BB'}$ be a pure state on $X\otimes X'\otimes B\otimes B'$, where $X'\simeq X$, $p_x \in [0,1]$ with $\sum_x p_x=1\,$, and the pure states $\ket{\xi_x}_{BB'}$ are arbitrary. 
Then
\begin{align}
&\widetilde{H}^{\downarrow}_{\alpha}(X|X'B)_{\rho}\leqslant \alpha \widebar{H}^{\downarrow}_{\alpha}(X|X'B)_{\rho} \quad \text{ for  $\alpha \in [0,1]$}  \label{eq:Bounds_ccq1}  \qquad \text{and}\\
& \widetilde{H}^{\uparrow}_{\alpha}(X|X'B)_{\rho}\leqslant \alpha \widebar{H}^{\uparrow}_{\alpha}(X|X'B)_{\rho}  \quad \text{ for  $\alpha \in [\tfrac{1}{2},1]$}  \label{eq:Bounds_ccq2} \, .
\end{align}
\end{myprop}
States $\rho_{XX'B}$ are sometimes called ``classically coherent'' as the classical information is treated coherently, i.e.\ fully quantum-mechanically. 
\begin{proof} [Proof of Proposition~\ref{ccq_states}]
It is known that $\widetilde{D}_1=\widebar{D}_1$ (see for example~\cite{tomamichel_quantum_2016}), and hence the claim is trivial in the case $\alpha=1$. Using the definition of $H^{\uparrow}_{\alpha}$ given in~\eqref{eq:cond_entropies} as well as the definitions of the Petz and the minimal quantum R\'enyi divergence given in~\eqref{eq:petz_entropy} and~\eqref{eq:min_entropy}, respectively\,, one can see that it suffices to show that
\begin{align}
&\widetilde{Q}_{\alpha}(\rho_{XX'B}\| \id_X\otimes \rho_{X'B})\leqslant \widebar{Q}_{\alpha}(\rho_{XX'B}\| \id_X\otimes \rho_{X'B})^{\alpha} \quad \text{ for  $\alpha \in (0,1)$}  \label{eq:Q_formulation_cq1},\\
&\widetilde{Q}_{\alpha}(\rho_{XX'B}\| \id_X\otimes \sigma_{X'B})\leqslant \widebar{Q}_{\alpha}(\rho_{XX'B}\| \id_X\otimes \sigma_{X'B})^{\alpha} \quad \text{ for  $\alpha \in [\tfrac{1}{2},1)$} , \label{eq:Q_formulation_cq2}
\end{align}
 for all density operators $\sigma_{X'B}$ (the case $\alpha=0$ then follows by continuity).

The marginal state $\rho_{X'B}$ appearing in \eqref{eq:Q_formulation_cq1} is a classical quantum (cq) state by assumption. Importantly, by the monotonicity of the R\'enyi divergence, we need only prove~\eqref{eq:Q_formulation_cq2} for cq states $\sigma_{X'B}$ in order to show~\eqref{eq:Bounds_ccq2}. Indeed, by Lemma~\ref{lem:DPI_dephasing} of Appendix~\ref{app:dephasing},  the supremum arising in equation~\eqref{eq:Bounds_ccq2} can be taken only over cq states.

Now define the unitary $U_{XX'}:=\sum_{x',x} \ket{x-x'}\bra{x}_X\otimes \ketbra{x'}_{X'}$, where arithmetic inside the ket is taken modulo $|X|$, and observe that $U_{XX'}\otimes \id_B$ leaves the state $\id_X \otimes \sigma_{X'B}$ invariant (here we use the assumption that  $\sigma_{X'B}$ is a cq state). 
Hence, by unitary invariance of $\mathbb{Q}_{\alpha}$, where $\mathbb{Q}_{\alpha}$ is a placeholder for $\widetilde{Q}$ or $\widebar{Q}$, we find
\begin{align}
&\mathbb{Q}_{\alpha}(  \rho_{XX'B}\|\id_X \otimes \sigma_{X'B}) \nonumber \\*
&\hspace{2em}=\mathbb{Q}_{\alpha}\big((U_{XX'}\otimes \id_B )\rho_{XX'B} (U_{XX'}^{*}\otimes \id_B) \| \id_X \otimes \sigma_{X'B}\big)\nonumber \\*
&\hspace{2em}= \mathbb{Q}_{\alpha}\big(\ketbra{0}_X\otimes \sum_{x,x'}\sqrt{p_xp_{x'}}\ket{x}\bra{x'}_{X'}\otimes \tr_{B'}\ket{\xi_x}\bra{\xi_{x'}}_{BB'} \|\id_X \otimes \sigma_{X'B}\big) \nonumber \\*
&\hspace{2em}=  \mathbb{Q}_{\alpha}\Big(\sum_{x,x'}\sqrt{p_xp_{x'}}\ket{x}\bra{x'}_{X'}\otimes \tr_{B'}\ket{\xi_x}\bra{\xi_{x'}}_{BB'} \|\sigma_{X'B}\Big) \nonumber \, ,
\end{align}
where we used the multiplicity of the trace under tensor products in the last equality. The claim now follows by a direct application of Corollary~\ref{cor:Upper_bound} (or more precisely of~\eqref{eq:reformulated_thm} applied to density operators):
\begin{align}
&\widetilde{Q}_{\alpha}(  \rho_{XX'B}\|\id_X \otimes \sigma_{X'B}) \nonumber\\*
&\hspace{6em}=  \widetilde{Q}_{\alpha}\Big(\sum_{x,x'}\sqrt{p_xp_{x'}}\ket{x}\bra{x'}_{X'} \otimes \tr_{B'}\ket{\xi_x}\bra{\xi_{x'}}_{BB'} \|\sigma_{X'B}\Big) \nonumber \\*
&\hspace{6em}\leqslant \widebar{Q}_{\alpha}\Big(\sum_{x,x'}\sqrt{p_xp_{x'}}\ket{x}\bra{x'}_{X'} \otimes \tr_{B'}\ket{\xi_x}\bra{\xi_{x'}}_{BB'} \|\sigma_{X'B}\Big)^{\alpha} \nonumber \\*
&\hspace{6em}=\widebar{Q}_{\alpha}(  \rho_{XX'B}\|\id_X \otimes \sigma_{X'B})^{\alpha} \,. \nonumber
\end{align}
This shows inequality~\eqref{eq:Q_formulation_cq2} for cq states $\sigma_{X'B}$, and hence~\eqref{eq:Bounds_ccq2}. Moreover, we recover inequality~\eqref{eq:Q_formulation_cq1} by setting $\sigma_{X'B}=\rho_{X'B}$.
 \end{proof}

\subsection{Relations between min-like entropies} \label{sec:min_like}

We can use duality relations for conditional entropies (see Lemma~\ref{lem_duality}) and Corollary~\ref{cor:entropies} to derive new bounds for conditional min-like entropies.

\begin{mylem} \label{lem:bounds_cond_entr}
For $\alpha \in [1,2]$ and $\rho_{AB}\in \cD(A\otimes B)\,$, we have that\footnote{We use the convention that $\frac{1}{0}=\infty\,$.}
\begin{align}
&\widetilde{H}^{\downarrow}_{\alpha}(A|B)_{\rho} \leqslant \alpha \widetilde{H}^{\uparrow}_{\frac{1}{2-\alpha}}(A|B)_{\rho}+(\alpha-1) \log |A| \label{eq:bounds_cond_entr1} \qquad \text{and}\\
&\widebar{H}^{\downarrow}_{\alpha}(A|B)_{\rho} \leqslant \frac{1}{2-\alpha} \left( \widebar{H}^{\uparrow}_{ \frac{1}{2-\alpha}}(A|B)_{\rho}+(\alpha-1) \log |A|\right). \label{eq:bounds_cond_entr2}
\end{align}
\end{mylem}

\begin{proof}
Let $\tau_{ABC}$ be a purification of $\rho_{AB}$ on $A \otimes B\otimes C$, i.e., $\tau_{ABC}$ is a pure state with $\tr_C \tau_{ABC}=\rho_{AB}$. 
Then, we find
\begin{align}
\widetilde{H}^{\downarrow}_{\alpha}(A|B)_{\tau}
=-\widebar{H}^{\uparrow}_{\frac{1}{\alpha}}(A|C)_{\tau} 
&\leq- \alpha \widetilde{H}^{\uparrow}_{\frac{1}{\alpha}}(A|C)_{\tau} +(\alpha-1) \log |A| \nonumber\\
&= \alpha \widetilde{H}^{\uparrow}_{\frac{1}{2-\alpha}}(A|B)_{\tau} +(\alpha-1) \log |A| \, , \nonumber
\end{align}
where we used  Corollary~\ref{cor:entropies} for the inequality and duality relations in the first and third equality. 
Similarly, we find
\begin{align}
\widebar{H}^{\downarrow}_{\alpha}(A|B)_{\tau}
&=-\widebar{H}^{\downarrow}_{2-\alpha}(A|C)_{\tau} \\
&\leqslant \frac{1}{2-\alpha} \left(-\widetilde{H}^{\downarrow}_{2-\alpha}(A|C)_{\tau} +(\alpha-1) \log |A|\right)\\
&=  \frac{1}{2-\alpha} \left(\widebar{H}^{\uparrow}_{\frac{1}{2-\alpha}}(A|B)_{\tau} +(\alpha-1) \log |A|\right) \, ,
\end{align}
where we again used Corollary~\ref{cor:entropies} for the inequality and duality relations in the first and third equality. 
\end{proof}

\begin{mycor} \label{cor:entropy_cq}
Let  $\alpha \in [1,2]$ and $\rho_{XB}$ be a cq state on $X \otimes B$, i.e., $\rho_{XB}=\sum_x p_x \ketbra{x}_X \otimes (\rho_x)_B$ where $(\rho_x)_B$ are density operators and $p_x \in [0,1]\,$, such that $\sum_x p_x=1\,$. Then
\begin{align}
&\widetilde{H}^{\downarrow}_{\alpha}(X|B)_{\rho} \leqslant \alpha \widetilde{H}^{\uparrow}_{\frac{1}{2-\alpha}}(X|B)_{\rho}  \label{eq:cond_entr_bound_cq1} \qquad \text{and} \\
&\widebar{H}^{\downarrow}_{\alpha}(X|B)_{\rho} \leqslant \frac{1}{2-\alpha}  \widebar{H}^{\uparrow}_{ \frac{1}{2-\alpha}}(X|B)_{\rho}.\label{eq:cond_entr_bound_cq2}
\end{align}
\end{mycor}

\begin{proof}
The proof proceeds analogously to the proof of Lemma~\ref{lem:bounds_cond_entr}, but we can make use of the improved bounds given in Proposition~\ref{ccq_states}: Let $\ket{\tau}_{XX'BB'}=\sum_x \sqrt{p_x}\ket{x}_X\ket{x}_{X'}\ket{\xi_x}_{BB'}$ where $\ket{\xi_x}_{BB'}$ purifies $(\rho_x)_B$. 
The system $X' \otimes B'$ corresponds to the system $C$ in the proof of Lemma~\ref{lem:bounds_cond_entr} and the state on $X \otimes X' \otimes B'$, i.e., $\tau_{XX'B'}$, is a classical-coherent state as required for Proposition~\ref{ccq_states} (note that the role of $B$ and $B'$ are interchanged here and in the statement of Proposition~\ref{ccq_states}). 
\end{proof}

We note that the special case $\alpha=2$ of the inequalities~\eqref{eq:bounds_cond_entr1}  and~\eqref{eq:cond_entr_bound_cq1} was already shown in~\cite{dupuis_entanglement_2015}.

\subsection{Equality condition for max-like  entropies} \label{sec:equality_cond_entropy}
In this section, we give a necessary and sufficient condition on a density operator $\rho_{AB}$, such that the entropies  $\widebar{H}^{\uparrow}_{\alpha}(A|B)_{\rho}$ and $\widetilde{H}^{\uparrow}_{\alpha}(A|B)_{\rho}$ are equal for $\alpha \in [\tfrac{1}{2},1)$. To derive the necessary condition, let $\alpha \in (0,1)$. In the proof of Lemma~1 of~\cite{tomamichel_relating_2014}, it is shown that the optimizer $\sigma^\star_B$ of $\widebar{H}^{\uparrow}_{\alpha}(A|B)_{\rho}=\sup_{\sigma_B \in \cD(B) }-\widebar{D}_{\alpha}(\rho_{AB}\| \id_A \otimes \sigma_B)$ is given by
 \begin{align} \label{def:tilde_sigma}
 \sigma^\star_B=\frac{\left(\tr_A \, \rho_{AB}^{\alpha}\right)^{\frac{1}{\alpha}}}{\tr\, \left(\tr_A{\rho_{AB}^{\alpha}}\right)^{\frac{1}{\alpha}}}\, .
 \end{align} 
 By the ALT inequality~\cite{lieb_inequalities_1976,araki_inequality_1990}, we then find that
  \begin{align}  \label{eq:ALT_for_entropies}
  \widebar{H}^{\uparrow}_{\alpha}(A|B)_{\rho}=-\widebar{D}_{\alpha}(\rho_{AB}\|\id_A \otimes \sigma^\star_B) \leqslant \sup \limits_{\sigma_B \in \cD(B) }-\widetilde{D}_{\alpha}(\rho_{AB}\| \id_A \otimes \sigma_B)
 \, .
 \end{align} 

According to Lemma~\ref{lem:eq_cond_ALT}, a necessary condition for equality in~\eqref{eq:ALT_for_entropies} is that  $[\rho_{AB},\id_A \otimes \sigma^\star_B]=0$. Assume now that $\alpha \in   [\tfrac{1}{2},1)$. To show that this condition is also sufficient for equality in~\eqref{eq:ALT_for_entropies}, it suffices to show that the function $\sigma_B \mapsto  -\widetilde{D}_{\alpha}(\rho_{AB}\|\id_A \otimes \sigma_B)$ or equivalently $\sigma_B \mapsto  \widetilde{Q}_{\alpha}(\rho_{AB}\|\id_A \otimes \sigma_B)$ attains its global maximum at $\sigma_B=\sigma^\star_B$ if $[\rho_{AB},\id_A \otimes \sigma^\star_B]=0$. The proof of this fact is based on standard derivative techniques, albeit for matrices, and is given in Appendix~\ref{app:equality_condition_extrema}. The results are summarized in the following lemma.
\begin{mylem}[Equality condition for entropies] \label{lem:equal_cond_entropies}
Let $\alpha \in [\tfrac{1}{2},1)\,$, $\rho_{AB}$ be a density operator and $\hat{\sigma}^\star_B:=\tr_A \, \rho_{AB}^{\alpha}$. Then, the following are equivalent
\begin{enumerate}
\item $\widebar{H}^{\uparrow}_{\alpha}(A|B)_{\rho}=\widetilde{H}^{\uparrow}_{\alpha}(A|B)_{\rho} $
\item $[\rho_{AB},\id_A \otimes \hat{\sigma}^\star_{B}]=0$ .
\end{enumerate}
\end{mylem}

\chapter{Pretty good measures in quantum information theory} \label{cha:pgm}

\section{Introduction}
In this section we present a  unified framework relating pretty good measures often used  in quantum information to their optimal counterparts. In particular, this can help to estimate the quality of a pretty good measure. In Section~\ref{sec:pgf}, we define the  ``pretty good fidelity'' as $\fpg(\rho ,\sigma):=\tr \sqrt{\rho}\sqrt{\sigma}$.  The new bound between the Petz and the minimal quantum R\'enyi  divergence given in Corollary~\ref{cor:Upper_bound} then implies that the pretty good fidelity is indeed pretty good in that $\fpg \leqslant F\leqslant \sqrt{\fpg}$, where $F$ denotes the usual fidelity defined by $F(\rho,\sigma):=\tr (\sqrt{\rho} \sigma \sqrt{\rho})^{\nicefrac{1}{2}}$.  Analogous bounds are also known between the pretty good guessing probability and the optimal guessing probability~\cite{barnum_reversing_2002} as well as between the pretty good and the optimal achievable singlet fraction~\cite{dupuis_entanglement_2015}.\footnote{Note that ``singlet'' refers to a maximally entangled state (and not necessarily to the maximally entangled two-qubit state)~\cite{dupuis_entanglement_2015}. }  We show that both of these relations follow by the inequality relating the pretty good fidelity and the fidelity. We thus present a unified picture of the relationship between pretty good quantities and their optimal versions.
Additionally, we show that the equality condition for the minimal and the Petz divergence given in Lemma~\ref{lem:eq_cond_ALT} lead to a new necessary and sufficient condition on the optimality of both pretty good measurement and singlet fraction.

\section{Pretty good fidelity} \label{sec:pgf}
Let   $\rho$ and $\sigma$ be two density operators throughout this section.  We define the pretty good fidelity of $\rho$ and $\sigma$ by
\begin{align}
\fpg(\rho,\sigma):=\widebar{Q}_{\frac{1}{2}}(\rho,\sigma)=\tr\,\sqrt\rho\sqrt\sigma \, .
\end{align}
This quantity was called the ``quantum affinity'' in \cite{luo_informational_2004} and is nothing but the fidelity of the ``pretty good purification'' introduced in~\cite{winter_extrinsic_2004}: Letting $\ket{\Omega}_{AA'}=\sum_k \ket{k}_A \ket{k}_{A'}$,  the canonical purification with respect to $\ket{\Omega}_{AA'}$ of $\rho$ is $\ket{\Psi_\rho}_{AA'}=({\sqrt{\rho}_A \otimes \id_{A'}}) \ket{\Omega}_{AA'}$, and thus 
\begin{align}
\fpg(\rho,\sigma)=\braket{\Psi_\rho|\Psi_\sigma}_{AA'}.
\end{align}

Recall that the usual fidelity is given by
\begin{align}
F(\rho,\sigma):=\widetilde{Q}_{\frac{1}{2}}(\rho,\sigma)
=\norm{\sqrt{\rho} \sqrt{\sigma}}_1
=\max_{V_{A'}}\bra{\Psi_\rho}(\id_{A}\otimes V_{A'})\ket{\Psi_\sigma}_{AA'}, 
\end{align}
where the maximum is taken over all unitary operators $V_{A'}$ and the final equality follows from Uhlmann's theorem~\cite{uhlmann_transition_1976}. Therefore, it is clear that $\fpg(\rho,\sigma)\leqslant F(\rho,\sigma)$. This can also be seen from the ALT inequality directly (cf. Corollary~\ref{cor:Upper_bound} for $\alpha=\tfrac 12$), and therefore, by Lemma~\ref{lem:eq_cond_ALT}, we have that $\fpg(\rho,\sigma)=F(\rho,\sigma)$ if and only if $[\rho, \sigma]=0$. 
The reverse ALT inequality implies a bound in the opposite direction; a similar approach using the H\"older inequality is given in \cite{audenaert_comparisons_2014}. 
By choosing $\alpha=\nicefrac{1}{2}$, it follows from Corollary~\ref{cor:Upper_bound} that the fidelity is also upper bounded by the square root of the pretty good fidelity, i.e.,
\begin{align} \label{ineq_pgf}
\fpg(\rho,\sigma)\leqslant F(\rho,\sigma)\leqslant \sqrt{\fpg(\rho,\sigma)}\, .
\end{align}
Hence the pretty good fidelity is indeed pretty good. 

Recall that the \emph{trace distance} between two density operators $\rho$ and $\sigma$ is defined by $\delta(\rho,\sigma):=\frac{1}{2}\|\rho-\sigma\|_1$.
An important property of the fidelity is its relation to the trace distance~\cite{fuchs_cryptographic_1999}:
\begin{align} \label{eq:fid_tr_dist}
1-F(\rho,\sigma) \leqslant \delta(\rho,\sigma) \leqslant \sqrt{1-F(\rho,\sigma)^2} \, . 
\end{align}
Indeed, the pretty good fidelity satisfies the same relation:
\begin{align} 
\label{eq:fvdg}
1-\fpg(\rho,\sigma) \leqslant \delta(\rho,\sigma) \leqslant \sqrt{1-\fpg(\rho,\sigma)^2} \, .
\end{align}
The upper bound follows immediately by combining the upper bound in \eqref{eq:fid_tr_dist} with the lower bound in~\eqref{ineq_pgf}. 
The lower bound was first shown in \cite{powers_free_1970} (see also \cite{audenaert_comparisons_2014}). 



\section{Relation to bounds for the pretty good measurement and singlet fraction} \label{sec:bounds_pg_measures}
In this section we show that together with entropy duality, the relation between fidelity and pretty good fidelity in \eqref{ineq_pgf} implies the known optimality bounds of the pretty good measurement and the pretty good singlet fraction. 
Let us first consider the optimal and pretty good singlet fraction. 
Define $R(A|B)_\rho$ to be the largest achievable overlap with the maximally entangled state one can obtain from $\rho_{AB}$ by applying a quantum channel on $B$. 
Formally, 
\begin{align} 
R(A|B)_{\rho}:=\text{max}_{\mathcal{E}_{B \rightarrow A'}} F(\ket{\Phi}\!\bra{\Phi}_{AA'}, (\id_A \otimes \mathcal{E}_{B \rightarrow A'})\rho_{AB})^2 \, ,
\end{align}
where $\ket{\Phi}_{AA'}=\frac{1}{\sqrt{|A|}}\sum_k \ket{k}_A \ket{k}_{A'}$ and the maximization is over all completely positive, trace-preserving maps ${\mathcal{E}_{B \rightarrow A'}}$. 
In~\cite{konig_operational_2009} it was shown that 
\begin{align} 
\widetilde{H}^{\uparrow}_{\infty}(A|B)_{\rho}=- \log |A|\, R(A|B)_{\rho} \, .
\end{align}
A ``pretty good'' map $\cE_{\text{pg}}$ was considered in \cite{berta_entanglement-assisted_2014}, and it was shown that  
\begin{align} 
\widetilde{H}^{\downarrow}_{2}(A|B)_{\rho}= -\log |A|\, R_{\text{pg}}(A|B)_{\rho} \, , 
\end{align}
where $R_{\text{pg}}(A|B)_\rho$ is the overlap obtained by using $\cE_{\text{pg}}$. 
Clearly $R_{\text{pg}}(A|B)_{\rho}\leqslant R(A|B)_{\rho}$, but the case $\alpha=2$ in~\eqref{eq:bounds_cond_entr1}, which comes from \eqref{ineq_pgf} via entropy duality, implies that we also have 
\begin{align} 
R_{\text{pg}}(A|B)_{\rho} \leqslant R(A|B)_{\rho}\leqslant \sqrt{R_{\text{pg}}(A|B)_{\rho}} \, .
\end{align}
This was also shown in~\cite{dupuis_entanglement_2015}. 
Note that in the special case where $\rho_{AB}$ has the form of a Choi state, i.e., $\tr_B \, \rho_{AB}=\frac{1}{|A|} \id_A$, this statement also follows from~\cite{barnum_reversing_2002}. 

Now let $\rho_{XB}=\sum_x p_x \ketbra{x}_X \otimes (\rho_x)_B$ be a cq state, and consider an observer with access to the system $B$ who would like to guess the variable $X$. 
Denote by $p_{\text{guess}}(X|B)$ the optimal guessing probability which can be achieved by performing a POVM on the system $B$.
It was shown in~\cite{konig_operational_2009} that 
\begin{align} 
\widetilde{H}^{\uparrow}_{\infty}(X|B)_{\rho} = -\log p_{\text{guess}}(X|B) \,  . 
\end{align}
On the other hand, it is also known that~\cite{buhrman_possibility_2008}
\begin{align} 
\widetilde{H}^{\downarrow}_{2}(X|B)_{\rho} = -\log p^{\text{pg}}_{\text{guess}}(X|B) \, ,
\end{align}
where $p_{\text{guess}}^{\text{pg}}(X|B)$ denotes the guessing probability of the pretty good measurement introduced in~\cite{belavkin_optimal_1975, hausladen_`pretty_1994}. (We refer to Appendix~\ref{app:pgm} for a description of the pretty good measurement.)
Clearly $p^{\text{pg}}_{\text{guess}} (X|B) \leqslant p_{\text{guess}}(X|B)$, but the case $\alpha=2$ in~\eqref{eq:cond_entr_bound_cq1}, which again comes from \eqref{ineq_pgf} via entropy duality, also implies that  
\begin{align}  \label{eq:pgm_is_pg}
p^{\text{pg}}_{\text{guess}} (X|B)  \leqslant p_{\text{guess}} (X|B)  \leqslant   \sqrt{p^{\text{pg}}_{\text{guess}} (X|B)}\, .  
\end{align}
This was originally shown in~\cite{barnum_reversing_2002}.

\section{Optimality conditions for pretty good measures} \label{sec:opt_cond_pgm}
Our framework also yields a novel optimality condition for the pretty good measures. 
Supposing $\tau_{ABC}$ is a purification of $\rho_{AB}$, the duality relations for R\'enyi entropies (cf. Lemma~\ref{lem_duality}) imply
\begin{align} \label{eq:aquival_dual_picture}
\widetilde{H}_{2}^{\downarrow}(A|B)_{\tau}=\widetilde{H}_{\infty}^{\uparrow}(A|B)_{\tau} \quad \iff \quad \widebar{H}^{\uparrow}_{\nicefrac{1}{2}}(A|C)_{\tau}= \widetilde{H}^{\uparrow}_{\nicefrac{1}{2}}(A|C)_{\tau}\, .
\end{align}
Applying the equality condition for max-like conditional entropies, using \linebreak Lemma~\ref{lem:equal_cond_entropies}, we find that the pretty good singlet fraction and pretty good measurement are optimal if and only if ${[\tau_{AC}, \id_A \otimes \hat{\sigma}^\star_{C}]}=0$, where $\hat{\sigma}^\star_{C}:=\tr_A \, \sqrt{ \tau_{AC}}$. 
Alternately, this specific equality condition ($\alpha=\nicefrac12$) can be established via weak duality of semidefinite programs, as described in Appendix~\ref{app:sdp}.

As a simple example of optimality of the pretty good singlet fraction, consider the case of a pure bipartite $\rho_{AB}$. 
Then every purification $\tau_{ABC}=\rho_{AB}\otimes \xi_C$ for some pure $\xi_C$.
Thus, $\tau_{AC}=\rho_A\otimes \xi_C$, and it follows immediately that the optimality condition is satisfied. 
Optimality also holds for arbitrary mixtures of pure states, i.e., for states of the form $\rho_{ABY}=\sum_y q_y \ketbra{\psi_y}_{AB}\otimes \ketbra {y}_Y$ with some arbitrary distribution $q_y$, provided both $B$ and $Y$ are used in the entanglement recovery operation. 
Here any purification takes the form $\ket{\tau}_{ABYY'}=\sum_y \sqrt{q_y}\ket{\psi_y}_{AB}\ket{y}_Y\ket{y}_{Y'}$. Hence, we have that $\tau_{AY'}=\sum_y q_y \tr_B \,\ketbra{\psi_y}_{AB} \otimes \ketbra{y}_{Y'}$, a state in which $Y'$ is classical, for which it is easy to see that the optimality condition holds. 

The optimality condition for the pretty good measurement can be simplified using the classical coherent nature of the state $\tau_{AC}$, which results in a condition formulated in terms of the Gram matrix. 
Suppose $\rho_{XB}=\sum_x p_x \ketbra{x}_X \otimes (\rho_x)_B$ describes the ensemble of mixed states $(\rho_x)_B$, for which a natural purification is given by 
\begin{align} \label{eq:purification_of_rho}
\ket{\tau}_{XX'BB'}=\sum_x \sqrt{p_x}\ket{x}_X\ket{x}_{X'}\ket{\xi_x}_{BB'}\, ,
\end{align}
where $\ket{\xi_x}_{BB'}$ denotes a purification of $(\rho_x)_B$.
Then we define the (generalized) Gram matrix $G$ 
\begin{align} \label{eq:G_def}
G_{X'B'}:=\sum_{x,x'}   \sqrt{p_xp_{x'}} \ket{x}\bra{x'}_{X'}  \otimes \tr_{B} \, \ket{\xi_{x}}\bra{\xi_{x'}}_{BB'} \, .
\end{align}
This definition reverts to the usual Gram matrix when the states $(\rho_x)_B$ are pure and system $B'$ is trivial. 
Observe that we are in the setting of Proposition~\ref{ccq_states}; using the unitary  $U_{XX'}$ introduced in its proof, we find that $\left(U_{XX'}\otimes \id_{B'}\right)\tau_{XX'B'}$ \linebreak $\left(U_{XX'}^{*}\otimes \id_{B'}\right)=\ketbra{0}_X\otimes G_{X'B'}$. 
Hence, $\sqrt{\tau_{XX'B'}}=\left(U_{XX'}^{*}\otimes \id_{B'}\right)(\ketbra{0}_X\otimes \sqrt{G_{X'B'}})\left(U_{XX'}\otimes \id_{B'}\right)$ and a further calculation shows that $\tr_X\sqrt{\tau_{XX'B'}}=\hat{\sigma}^\star_{X'B'}$, with 
\begin{align}
\hat{\sigma}^\star_{X'B'}:=\sum_x \ketbra{x}_{X'} \otimes \bra{x} \sqrt{G_{X'B'}}\ket{x}_{X'}.
\end{align}
Note that $[M,N]=0$ is equivalent to $[UMU^{*},UNU^{*}]=0$ for any square matrices $M,N$ and any unitary $U$. Therefore, we find that the equality condition $[\tau_{XX'B'}, \id_X \otimes \hat{\sigma}^\star_{X'B'}]=0$ is equivalent to $[\ketbra{0}_X \otimes G_{X'B'}, \id_X \otimes\hat{\sigma}^\star_{X'B'}]=[G_{X'B'},\hat{\sigma}^\star_{X'B'}]=0$. Thus we have shown the following result:
\begin{mylem}[Optimality condition for the pretty good measurement] \label{lem:opt_pgm}
The pretty good measurement is optimal for distinguishing states in the ensemble $\{p_x,\rho_x\}$ if and only if $[G_{X'B'},\hat{\sigma}^\star_{X'B'}]=0$.
\end{mylem}

In the case of distinguishing pure states, we recover Theorem~2 of~\cite{dalla_pozza_optimality_2015} (which was first shown in~\cite{helstrom_bayes-cost_1982}). 
To see this, observe that $B'$ is now trivial and $G_{X'}$ is the usual Gram matrix. 
Moreover, $\hat{\sigma}^\star_{X'}$ is now the diagonal of the square root of $G_{X'}$, and the commutation condition of Lemma~\ref{lem:opt_pgm} becomes $[G_{X'},\hat{\sigma}^\star_{X'}]=0$, which is equivalent to the condition in equation~(11) of~\cite{dalla_pozza_optimality_2015} (in the case of the pretty good measurement). 
Reformulating what it means for the Gram matrix $G_{X'}$ to commute with the diagonal matrix $\hat{\sigma}^\star_{X'}$ then leads to Theorem~3 of~\cite{dalla_pozza_optimality_2015}.

\section{Conclusion}

The bound between the minimal and the Petz divergence given in Corollary~\ref{cor:Upper_bound} leads to an elegant unified framework of pretty good constructions in quantum information theory, and the ALT equality condition leads to a simple necessary and sufficient condition for their optimality. 
Previously it was observed that the min entropy $\widetilde H^{\uparrow}_\infty$ characterizes optimal measurement and singlet fraction, while $\widetilde H_2^\downarrow$ is the ``pretty good min entropy'' since it characterizes pretty good measurement and singlet fraction. 
On the other hand, we can think of $\widebar H_{\nicefrac{1}{2}}^\uparrow$ as the ``pretty good max entropy'' since it is based on the pretty good fidelity instead of the (usual) fidelity itself as in the max entropy $\widetilde H_{\nicefrac{1}{2}}^{\uparrow}$. 
Entropy duality then beautifully links the two, as the (pretty good) max entropy is dual to the (pretty good) min entropy, and the known optimality bounds can be seen to stem from the lower bound on the pretty good fidelity in \eqref{ineq_pgf}. 
Indeed, that such a unified picture might be possible was the original inspriation to look for a reverse ALT inequality of the form given in Theorem~\ref{thm:Inv_ALT}. 
It is also interesting to note that both the pretty good min and max entropies appear in achievability proofs of information processing tasks, the former in randomness extraction against quantum adversaries~\cite{tomamichel_leftover_2011} and the latter in the data compression with quantum side information~\cite{renes_one-shot_2012}.

%
%
%
\clearpage
%
%
%
\begin{appendix}
%
%
\lhead[\fancyplain{\scshape Appendix \thechapter}
{\scshape Appendix \thechapter}]
{\fancyplain{\scshape \leftmark}
  {\rightmark}}
\rhead[\fancyplain{\scshape \leftmark}
{\scshape \leftmark}]
{\fancyplain{\scshape Appendix \thechapter}
  {\scshape Appendix \thechapter}}
%
%

\chapter{Pretty good measurement} \label{app:pgm}

In this appendix, we give some additional information about the pretty good measurement, completing the discussion in the preface. Let us first formalize our goal: \\
Fix a set of density operators $\{\rho_x\}$ on a quantum system $B$ and a discrete probability distribution $P_X$ with finite support. Alice chooses an $x$ with probability  $P_X(x)=:p_x$ and prepares the corresponding state $\rho_x$ on the system $B$. Bob has access to system $B$ and wants to find out which $x$ has been chosen by Alice. We can summarize the information from the point of view of Bob in the following cq state $\rho_{XB}=\sum_{x} p_x \ketbra{x}_X \otimes (\rho_x)_B$. The measurement preformed by Bob can be described by POVM elements $\Lambda:=\{\Lambda_x\}$ on the system $B$. Then, the probability that Bob guesses correctly in the case that Alice has chosen $x$ is given by $\tr \, \Lambda_x \rho_x$, and hence, the unconditioned success probability (using the POVM $\Lambda$) is $p^{\Lambda}_{\text{guess}}(X|B):=\sum_{x} P_X(x) \, \tr \, \Lambda_x \rho_x$. Therefore, our goal is to find the POVM elements $\Lambda_x$ that maximize $p^{\Lambda}_{\text{guess}}(X|B)$. We define
\begin{align}
p_{\text{guess}}(X|B):= \underset{\Lambda_x}{\max} \sum_{x} P_X(x) \, \tr \, \Lambda_x \rho_x \,.
\end{align}
Unfortunately, it turns out that this optimization problem is not easy to solve in general. However, a different approach was taken in~\cite{belavkin_optimal_1975, hausladen_`pretty_1994}. Indeed,  they defined the pretty good POVM elements 
\begin{align}
\Lambda^{\text{pg}}_x:=P_X(x) \, \hat{\rho}^{-\frac{1}{2}} \rho_x \hat{\rho}^{-\frac{1}{2}} \, ,
\end{align}
where we set $ \hat{\rho}:= \sum_{x} P_X(x) \, \rho_x\,$. Then, the pretty good success probability is given by 
\begin{align}
p^{\text{pg}}_{\text{guess}}(X|B):=\sum_{x} P_X(x) \, \tr \, \Lambda^{\text{pg}}_x \rho_x \, .  
\end{align}
It turns out that the choice $\Lambda_x=\Lambda^{\text{pg}}_x$ is indeed pretty good in that $p_{\text{guess}}(X|B)$ is bounded from below and above in terms of $p^{\text{pg}}_{\text{guess}}(X|B)$ (cf.~\eqref{eq:pgm_is_pg} for the exact statement). These bounds follow elegantly in the framework of this thesis as discussed in detail in Chapter~\ref{cha:pgm}.


 \chapter{Technical results} \label{app:technical_res}

\section{Optimal marginals for classically coherent states} \label{app:dephasing}
This appendix details the argument that cq states are optimal in the conditional entropy expressions for classically coherent states. 
Following the approach taken in \cite[Lemma A.1]{dupuis_decoupling_2014} to establish a similar result for the smooth min entropy, we can show the following lemma.
\begin{mylem} \label{lem:DPI_dephasing}
Let   $ \ket{\rho}_{XX'BB'}=\sum_x \sqrt{p_x}\ket{x}_X\ket{x}_{X'}\ket{\xi_x}_{BB'}$ be a pure state on  $X\otimes X' \otimes B \otimes B'$, where $p_x \in [0,1]$ with $\sum_x p_x=1\,$, and $X'\simeq X$. Let $\mathbb{Q}_{\alpha}$ be a placeholder for $\widetilde{Q}_{\alpha}$ or $\widebar{Q}_{\alpha}$. Then, for any density operator $\sigma_{X'B}\,$, we have that 
 \begin{align}
 \label{eq:classicalcoherentQ}
&\mathbb{Q}_{\alpha}(\rho_{XX'B}\| \id_X \otimes \sigma_{X'B})\leqslant \mathbb{Q}_{\alpha}(\rho_{XX'B}\| \id_X \otimes \sigma^{\cl}_{X'B}) \quad \text{for }  \alpha \in [\tfrac{1}{2},1)   \, ,
 \end{align}
 where $\sigma^{\cl}_{X'B}:=\sum_{x}\ketbra x_{X'}\otimes \bra{x}\sigma_{X'B}\ket{x}_{X'}$. 
\end{mylem}
\begin{proof} 
Let $P_{XX'}=\sum_x \ketbra{x}_{X} \otimes \ketbra{x}_{X'}$ and define the quantum channel $\mathcal E$ from $X \otimes X'$ to itself by $\mathcal{E}(\cdot):= P_{XX'}(\cdot)P_{XX'}+(\id_{XX'}-P_{XX'})(\cdot)(\id_{XX'}-P_{XX'})$. 
Since $P_{XX'}\ket{\Psi}_{XX'BB'}=\ket{\Psi}_{XX'BB'}$, $\mathcal{E}_{XX'}\otimes \mathcal I_B$ leaves the density operator $\rho_{XX'B}$ invariant.
By the DPI we then have, for  $\alpha \in [\tfrac{1}{2},1)$,
 \begin{align}
&\mathbb{Q}_{\alpha}(\rho_{XX'B}\| \id_X \otimes \sigma_{X'B})\\
&\hspace{5em}\leqslant \mathbb{Q}_{\alpha}\big(\rho_{XX'B}\| \mathcal{E}_{XX'} \otimes \mathcal I_B(\id_X \otimes \sigma_{X'B})\big) \\
&\hspace{5em}=\mathbb{Q}_{\alpha}\big(\rho_{XX'B}\| (P_{XX'} \otimes \id_B)(\id_X \otimes \sigma_{X'B})(P_{XX'} \otimes \id_B)\big)\,. \label{eq:indiff_supp}
 \end{align}
In~\eqref{eq:indiff_supp} we use the fact that $\mathbb{Q}_\alpha$ is indifferent to parts of its second argument which are not contained in the support of its first argument. 
Observe that $(P_{XX'} \otimes \id_B)(\id_X \otimes \sigma_{X'B})({P_{XX'} \otimes \id_B})=\sum_x \ketbra x_{X} \otimes \ketbra{x}_{X'}\sigma_{X'B}\ketbra{x}_{X'}\leqslant \id_{X}  \otimes \sigma^{\cl}_{X'B}$. Inequality~\eqref{eq:classicalcoherentQ} now follows directly from the dominance property of $\mathbb{D}_{\alpha}$ (see e.g.,~\cite{tomamichel_quantum_2016}), which states (in terms of $\mathbb{Q}_{\alpha}$) that $\mathbb{Q}_{\alpha}(\rho\|  \sigma) \leqslant \mathbb{Q}_{\alpha}(\rho\|  \sigma')$ for any non-negative operators $\rho, \sigma, \sigma'$ with $\sigma \leqslant \sigma'$ .
\end{proof}

\section{Sufficient condition for equality of max-like entropies} \label{app:equality_condition_extrema}

In this appendix, we show that, for $\alpha \in [\tfrac{1}{2},1)$, the function  $f_{\alpha}: \cD(B) \ni \sigma_B \mapsto  \widetilde{Q}_{\alpha}(\rho_{AB}\|\id_A \otimes \sigma_B)$  attains its global maximum at $\sigma_B=\sigma^\star_B$  if $[\rho_{AB},\id_A \otimes \sigma^\star_B]=0$. We use the notation of Section~\ref{sec:equality_cond_entropy}. The following lemma is similar to Lemma~5.1 of~\cite{sebastiani_derivatives_1996-1}. 

\begin{mylem} \label{lem:derivative}
Let $I \subset \mathbb{R}$ be open and  $t_0 \in I$. Let $A(t)$ be a matrix whose entries are smooth functions of $t \in I$ and $A(t)>0$ for all $t \in I$. 
Further, let $B$ be a matrix such that $[B,A(t_0)]=0$. Then,
\begin{align} \label{eq:derivative_of_matrix_power}
\frac{d}{dt}\Bigr|_{\substack{t=t_0}} \tr \, BA(t)^r=r \, \tr \, B A(t_0)^{r-1} A'(t_0)  \quad \text{for} \quad r \in \mathbb{R} \, ,
\end{align}
where $A'(t_0):=\frac{d}{dt}\Bigr|_{\substack{t=t_0}} A(t)$\, .
\end{mylem}
\begin{proof}
Note that it is straightforward to adapt Theorem 3.5 of~\cite{sebastiani_derivatives_1996-1} to the complex case. Therefore, by setting $\alpha=0$ in the equation~(26) of~\cite{sebastiani_derivatives_1996-1}, we find that
\begin{align}
\frac{d}{dt}\Bigr|_{\substack{t=t_0}} \tr \,  BA(t)^r =r \,  \tr   \, B  A'(t_0) A(t_0)^{r-1} + r \,  \tr \,  B H_{0,r} A(t_0)^{r-1}  \, ,
\end{align}
where  $H_{0,r}$ is defined in equation~(27) of~\cite{sebastiani_derivatives_1996-1}. Since $[A(t_0),B]=0$, a short calculation shows that $\tr \, B H_{0,r} A(t_0)^{r-1}=0$.
\end{proof}
\begin{mylem} \label{lem:derivative_tilde_Q}
Set $I=(-\delta,\delta) \subset \mathbb{R}$ for some $\delta>0$ and let $A(t)$ be a matrix whose entries are smooth functions of $t \in I$ and $A(t)>0$ for all $t \in I$.  For $B$ a density operator such that $[B,A(0)]=0$, 
\begin{align}
\frac{d}{dt}\Bigr|_{\substack{t=0}} \widetilde{Q}_{\alpha}(B\|A(t))=(1-\alpha) \, \Re \, \tr \, B^{\alpha} A(0)^{-\alpha} A'(0)  \quad \text{for} \quad \alpha \in (0,1) \, ,
\end{align}
where $A'(t_0):=\frac{d}{dt}\Bigr|_{\substack{t=t_0}} A(t)$ for $t_0 \in I$.
\end{mylem}
\begin{proof}
To simplify the notation, let us define $\beta:=\tfrac{1-\alpha}{2\alpha}$. We set $B_{\varepsilon}:=B+\varepsilon \, \id>0$ for some $\varepsilon>0$. Using Lemma~\ref{lem:derivative} (with $A=A(t)^{\beta}B_{\varepsilon}A(t)^{\beta}$ and $B=\id$), we~find  
\begin{align}
\frac{d}{dt}\Bigr|_{\substack{t=t_0}}& \tr \left(A(t)^{\beta}B_{\varepsilon}A(t)^{\beta}\right)^{\alpha} \nonumber \\*
&= \alpha \, \tr \, \left(A(t_0)^{\beta}B_{\varepsilon}A(t_0)^{\beta}\right)^{\alpha-1} \frac{d}{dt}\Bigr|_{\substack{t=t_0}} \left(A(t)^{\beta}B_{\varepsilon}A(t)^{\beta} \right)  \nonumber \\*
&= \alpha \,  \tr \, \left(A(t_0)^{\beta}B_{\varepsilon}A(t_0)^{\beta}\right)^{\alpha-1} \Bigg( \frac{d}{dt}\Bigr|_{\substack{t=t_0}} A(t)^{\beta}B_{\varepsilon}A(t_0)^{\beta}  \nonumber \\*
&\hspace{16em}+A(t_0)^{\beta} B_{\varepsilon} \frac{d}{dt}\Bigr|_{\substack{t=t_0}}A(t)^{\beta} \Bigg) \, .  \nonumber 
\end{align} 
This can be simplified by noting that for any Hermitian matrix $H$ and any matrix~$C$, 
\begin{align}
\tr  \, H(C+C^{*})
&=\tr \, HC+\tr \, H C^{*}   \nonumber \\*
&= \tr \, HC+\tr \, H^{*}  C^{*}  \nonumber \\*
&=\tr \, HC+\left(\tr \, HC \right)^{\ast}   \nonumber \\*
&=2 \, \text{Re} \, \tr \, HC \, . \nonumber
\end{align} 
Using this we obtain
\begin{align} \label{eq:derivative_Q_calculation}
&\frac{d}{dt}\Bigr|_{\substack{t=t_0}} \widetilde{Q}_{\alpha}\big(B_{\varepsilon}\|A(t)\big) \nonumber\\*
&\hspace{4em}=2 \alpha \, \Re \, \tr  \, \left(A(t_0)^{\beta}B_{\varepsilon}A(t_0)^{\beta}\right)^{\alpha-1} \frac{d}{dt}\Bigr|_{\substack{t=t_0}} A(t)^{\beta}B_{\varepsilon}A(t_0)^{\beta} \nonumber \\*
&\hspace{4em}=2 \alpha \, \Re \, \tr \, A(t_0)^{-\beta} \left(A(t_0)^{\beta}B_{\varepsilon} A(t_0)^{\beta}\right)^{\alpha} \frac{d}{dt}\Bigr|_{\substack{t=t_0}} A(t)^{\beta}\, . \nonumber
\end{align} 
Taking the limit $\varepsilon \rightarrow 0$ yields
\begin{align}
\lim_{\varepsilon\to 0}\frac{d}{dt}\Bigr|_{\substack{t=t_0}} \widetilde{Q}_{\alpha}\big(B_{\varepsilon}\|A(t)\big)
&=2 \alpha \, \Re \, \tr \, A(t_0)^{-\beta} \left(A(t_0)^{\beta}B A(t_0)^{\beta}\right)^{\alpha} \frac{d}{dt}\Bigr|_{\substack{t=t_0}} A(t)^{\beta}\, . \nonumber
\end{align}
At $t_0=0$ the righthand side can be simplified by again making use of Lemma~\ref{lem:derivative} as well as $[A(0),B]=0$:
\begin{align}
\lim_{\varepsilon\to 0}\frac{d}{dt}\Bigr|_{\substack{t=0}} \widetilde{Q}_{\alpha}\big(B_{\varepsilon}\|A(t)\big)
&=2 \alpha \, \Re \, \tr  \, B^{\alpha} A(0)^{\beta (2\alpha-1)}\frac{d}{dt}\Bigr|_{\substack{t=0}} A(t)^{\beta} \nonumber \\
&=(1-\alpha) \, \Re  \,  \tr \,  B^{\alpha} A(0)^{-\alpha} A'(0) \, . \nonumber
\end{align}

It remains to be shown that the limit can be interchanged with the derivative. 
This follows if we ensure that $\frac{d}{dt}\bigr|_{\substack{t=t_0}} \widetilde{Q}_{\alpha}(B_{\varepsilon}\|A(t))$ converges uniformly in $t_0 \in [-\nicefrac{\delta}{2},\nicefrac{\delta}{2}]$ for $\varepsilon \rightarrow 0$. 
To show uniform convergence, it suffices to show  
\begin{align}
&\lim_{\varepsilon \rightarrow 0} \sup_{t_0 \in  [-\delta/2,\delta/2]} \Bigg\| A(t_0)^{-\beta} \Big[ \left(A(t_0)^{\beta}B_{\varepsilon} A(t_0)^{\beta}\right)^{\alpha} \nonumber \\
&\hspace{12em}-\left(A(t_0)^{\beta}B A(t_0)^{\beta}\right)^{\alpha} \Big] \frac{d}{dt}\Bigr|_{\substack{t=t_0}} A(t)^{\beta}\Bigg\|_1 =0\, , \nonumber
\end{align}
where we used that $|\tr (M)|\leqslant \norm{M}_1$ for any square matrix $M$ (see, e.g.,~\cite[Exercise~IV~2.12]{bhatia_matrix_1997}). By the generalized H\"older inequality for matrices (cf. Corollary~\ref{cor:generalized_Holder}), we find that it is enough to show that
\begin{align}
&\lim_{\varepsilon \rightarrow 0} \sup_{t_0 \in  [-\delta/2,\delta/2]} \norm{ \left(A(t_0)^{\beta}B_{\varepsilon} A(t_0)^{\beta}\right)^{\alpha}-\left(A(t_0)^{\beta}B A(t_0)^{\beta}\right)^{\alpha} }_1\nonumber \\*
&\hspace{17.5em}\norm{A(t_0)^{-\beta}}_{\infty}  \norm{ \frac{d}{dt}\Bigr|_{\substack{t=t_0}} A(t)^{\beta}}_{\infty} =0  \nonumber \, . 
\end{align}
Note that the infinity-norm terms are bounded on the compact interval \linebreak $t_0 \in  [-\nicefrac{\delta}{2},\nicefrac{\delta}{2}]$, as $A(t)^{\beta}$ is continuously differentiable for $A(t)>0$.
Thus, we need only show  that
\begin{equation}\label{eq:unif_conv_matrix_powers}
\lim_{\varepsilon \rightarrow 0} \sup_{t_0 \in  [-\nicefrac{\delta}{2},\nicefrac{\delta}{2}]}  \norm{ \left(A(t_0)^{\beta}B_{\varepsilon} A(t_0)^{\beta}\right)^{\alpha}-\left(A(t_0)^{\beta}B A(t_0)^{\beta}\right)^{\alpha} }_1=0\, .
\end{equation}
Since $t \rightarrow t^{\alpha}$ is operator monotone for $\alpha \in [0,1]$ (L\"owner's theorem~\cite{Lowner1934}), the matrix inside the trace norm is positive, and hence~\eqref{eq:unif_conv_matrix_powers} is equivalent to
\begin{align}\label{eq:unif_conv_matrix_powers2}
\lim_{\varepsilon \rightarrow 0} \sup_{t_0 \in  [-\nicefrac{\delta}{2},\nicefrac{\delta}{2}]} \norm{ \left(A(t_0)^{\beta}B_{\varepsilon} A(t_0)^{\beta}\right)^{\alpha}}_1- \norm{\left(A(t_0)^{\beta}B A(t_0)^{\beta}\right)^{\alpha} }_1=0\, .
\end{align}
Note that $\varepsilon \mapsto \|(A(t_0)^{\beta}B_{\varepsilon} A(t_0)^{\beta})^{\alpha}\|_1$ is monotonically decreasing (again by  L\"owner's theorem). Then, by Dini's theorem,  it converges uniformly to  \linebreak $\|(A(t_0)^{\beta}B A(t_0)^{\beta})^{\alpha}\|_1$, which proves~\eqref{eq:unif_conv_matrix_powers2}, and hence the desired uniformity of the convergence. 
\end{proof}
We are now ready to calculate the derivative of the function  $f_{\alpha}$ at $\sigma_B=\sigma^\star_B$.
\begin{mylem}
\label{lem:alphasufficiency}
Let $\alpha \in [\tfrac{1}{2},1)$ and $\rho_{AB} \in \cD(A\otimes B)$ be such that $[\rho_{AB}, \id_A \otimes \sigma^\star_B]=0$. Then the  function $f_{\alpha}:\cD(B) \ni \sigma_B\mapsto  \widetilde{Q}_{\alpha}(\rho_{AB}\|\id_A \otimes \sigma_B)$ attains its global maximum at $\sigma^\star_B$ as defined in \eqref{def:tilde_sigma}. 
\end{mylem}
\begin{proof} 
First consider the case $\rho_{AB}>0$ for simplicity; we return to the rank-deficient case below. 
Since $(\rho,\sigma) \mapsto \widetilde{Q}_{\alpha}(\rho\| \sigma)$ is jointly concave~\cite{frank_monotonicity_2013, beigi_sandwiched_2013}, the function $ f_{\alpha} :  \mathcal{D}(B)\ni \sigma_B \mapsto {\widetilde Q_\alpha(\rho_{AB} \| \id_A \otimes \sigma_B)}$ is concave. As $ \mathcal{D}(B)$ is a convex set, it suffices to show that $f_\alpha$ has an extreme point at $\sigma^\star_B$ (which is then also a global maximum). 
Observe that $\sigma^\star_B>0$ by definition, and therefore all states $\sigma_B(t)$ along arbitrary paths of states through $\sigma_B(0)=\sigma^\star_B$ have full rank for all $t$ sufficiently close to zero.  
Thus, we may use Lemma~\ref{lem:derivative_tilde_Q} to compute the derivative along any such path and find  
\begin{align*}
&\frac{d}{dt}\Bigr|_{\substack{t=0}}  \widetilde{Q}_{\alpha}(\rho_{AB}\|\id_A \otimes \sigma_B(t))\\*
&\hspace{5em}= (1-\alpha) \, \Re \,  \tr \, \rho_{AB}^{\alpha} \left(\id_A \otimes \sigma^\star_B\right)^{-\alpha}  \left(\id_A \otimes \frac{d}{dt}\Bigr|_{\substack{t=0}} \sigma_{B}(t) \right)\\*
&\hspace{5em}= (1-\alpha) \, \Re \, \tr  \left( \tr_A \left(\rho_{AB}^{\alpha}\right) (\sigma^\star_B)^{-\alpha}  \frac{d}{dt}\Bigr|_{\substack{t=0}} \sigma_{B}(t)  \right)  
=0 \, .
\end{align*} 
Therefore $\sigma^\star_B$ is the optimizer in this case. 

For $\rho_{AB}$ not strictly positive, we can restrict the set of marginal states $\sigma_B$ to the support of $\sigma^\star_B$ and replay the above argument. 
To see this, first observe that the support of $\sigma^\star_B$ is the same as that of $\rho_{B}$. 
Furthermore, as noted in~\cite{muller-lennert_quantum_2013}, the DPI for $\widetilde{D}_{\alpha}$ implies that the maximum of  $f_{\alpha}$ is always attained at a density matrix $\sigma^{\star}_B$ satisfying $\sigma^{\star}_B\ll \rho_B$.
Therefore, we can restrict the domain of the function $f_{\alpha}$ to the set $\mathcal{P}(B):=\{\sigma_B \in \cD(B): \sigma_B \ll \sigma^\star_B\}$. 
Now observe that $\ker (\id_A\otimes\sigma^\star_B)\subseteq \ker (\rho_{AB})$.
For any $\ket{\psi}_B$ we have $\bra{\psi}\rho_B\ket{\psi}_B=\sum_k \bra{k}_A \bra{\psi} _B \rho_{AB}\ket{k}_A \ket{\psi}_B$. 
By positivity of $\rho_{AB}\geqslant 0$, each $\ket{\psi}_B\in \ker(\sigma^\star_B)=\ker(\rho_B)$ leads to a set of states $\ket{k}_A\otimes \ket{\psi}_B\in \ker(\rho_{AB})$.
This implies that projecting  $\rho_{AB}$ to the support of $\id_A\otimes \sigma^\star_B$ has no effect on $\widetilde Q_{\alpha}$. 
Hence, we can restrict all operators in the problem to this subspace, where again all states in $\mathcal P(B)$ sufficiently close to $\sigma^\star_B$ have full rank.
\end{proof}

\section{Optimality condition via semidefinite programming}
\label{app:sdp}

Here we derive the optimality condition for pretty good measures via weak duality of semidefinite programs. 
In terms of fidelity and pretty good fidelity, the optimality condition in \eqref{eq:aquival_dual_picture} reads 
\begin{align}
\label{eq:fpgopt}
\fpg(\tau_{AC},\id_A\otimes \sigma^\star_C)=\sup_{\sigma\in \cD(C)} F(\tau_{AC},\id_A\otimes \sigma_C),
\end{align}
where $ \sigma^\star_C$ is as in \eqref{def:tilde_sigma} with $\alpha=\nicefrac 12$. 
Lemma~\ref{lem:eq_cond_ALT} implies that $[\tau_{AC},\id_A\otimes \sigma^\star_C]=0$ is necessary for \eqref{eq:fpgopt} to hold. 
Sufficiency, meanwhile, is the statement that $\sigma^\star_C$ is the optimizer on the righthand side. 
We can show this by formulating the optimization as a semidefinite program and finding a matching upper bound using the dual program. 

In particular, following~\cite{watrous_semidefinite_2009}, the optimal value of the (primal) semidefinite program
\begin{align}
\begin{array}{r@{\,\,}rl}
\gamma=&\sup & \tr \, W_{ACA'C'}\tau_{ACA'C'}\\
&\text{s.t.} & \tr_{A'C'} W_{ACA'C'}\leqslant \id_A\otimes \sigma_C\\
&& \tr \,\sigma_C\leqslant 1\\
&& W_{ACA'C'},\sigma_C \geqslant 0 \, ,
\end{array}
\end{align}
satisfies $\gamma=\sup_{\sigma\in \cD(C)} F(\tau_{AC},\id_A\otimes \sigma_C)^2$.
Here $A'\simeq A$, $C'\simeq C$, and we take $\tau_{ACA'C'}$ to be the canonical purification of $\tau_{AC}$ as in Section~\ref{sec:pgf}. 
Using Watrous's general form for semidefinite programs we can easily derive the dual, which turns out to be  
\begin{align}
\begin{array}{r@{\,\,}rl}
\beta=&\inf & \mu \\
&\text{s.t.} & Z_{AC}\otimes \id_{A'C'}\geqslant \tau_{ACA'C'}\\
&&\mu\id_C\geq\tr_A Z_{AC}\\
&& \mu,Z_{AC}\geqslant 0 \, .
\end{array}
\end{align}


By weak duality $\gamma\leqslant \beta$, but the following choice of $\mu$ and $Z_{AC}$ gives $\beta=\fpg(\tau_{AC},\id_A\otimes \sigma^\star_C)^2$ and therefore \eqref{eq:fpgopt}:
\begin{align}
&\mu^\star = \left(\tr\, \sqrt{\tau_{AC}}\sqrt{\id_A\otimes \sigma^\star_C}\right)^2\quad \text{and} \\
&Z_{AC}^\star=\tr\left(\sqrt{\tau_{AC}}\sqrt{\id_A\otimes \sigma^\star_C}\right){\tau_{AC}^{\nicefrac12}}\left(\id_A\otimes \sigma^\star_C\right)^{-\nicefrac12}.
\end{align}
Here the inverse of $\id_A\otimes\sigma^\star_C$ is taken on its support.  To see that the first feasibility constraint is satisfied, start with the operator inequality 
\begin{align*}
&\id_{ACA'C'} \tr \, \sqrt{\tau_{AC}}\sqrt{\id_A\otimes \sigma^\star_C}\geqslant \left(\tau_{AC}^{\nicefrac14}\left(\id_A\otimes \sigma^\star_C\right)^{\nicefrac14} \otimes \id_{A'C'} \right)\Omega_{ACA'C'}  \\
&\hspace{20em}\left(\tau_{AC}^{\nicefrac14}(\id_A\otimes \sigma^\star_C)^{\nicefrac14} \otimes \id_{A'C'} \right),
\end{align*}
which holds because the righthand side is the canonical purification of the positive operator $\sqrt{\tau_{AC}}\sqrt{\id_A\otimes \sigma^\star_C}$ and the trace factor on the left is its normalization. 
Conjugating both sides by $\tau_{AC}^{\nicefrac14}(\id_A\otimes \sigma^\star_C)^{-\nicefrac14} \otimes \id_{A'C'}$ preserves the positivity ordering and gives 
\begin{align}\label{eq:feasibility_constraint}
&\tr\left(\sqrt{\tau_{AC}}\sqrt{\id_A\otimes \sigma^\star_C}\right) \tau_{AC}^{\nicefrac12}(\id_A\otimes \sigma^\star_C)^{-\nicefrac12}\otimes \id_{A'C'} \nonumber \\
&\hspace{10em}\geqslant \left( \tau_{AC}^{\nicefrac12} \otimes \id_{A'C'} \right)\Omega_{ACA'C'} \left( \tau_{AC}^{\nicefrac12}\otimes \id_{A'C'} \right) \, ,
\end{align}
where we used that $\ker \left(\id_A\otimes\sigma^\star_C \right) \subseteq \ker (\tau_{AC})$ (just as in the proof of \linebreak Lemma~\ref{lem:alphasufficiency}), ensuring that $\tau_{AC}(\id_A\otimes \sigma^\star_C)^{-1} (\id_A\otimes \sigma^\star_C)=\tau_{AC}$. Note that inequality~\eqref{eq:feasibility_constraint} shows that $Z^\star_{AC}\otimes \id_{A'C'}\geqslant \tau_{ACA'C'}$. Meanwhile, the second constraint is satisfied (with equality in the case where  $\sigma^\star_C$ has full rank) because direct calculation shows that $\tr_A {\tau_{AC}^{\nicefrac12}}(\id_A\otimes \sigma^\star_C)^{-\nicefrac12}\leq\id_C \, \tr\, \sqrt{\tau_{AC}}\sqrt{\id_A\otimes \sigma^\star_C}\,$.

\chapter{Notation and abbreviations} \label{APPnoation}

For an overview of the notation for quantum R\'enyi divergences and quantum conditional R\'enyi entropies used in this thesis, see Section~\ref{sec:important_divergences} and  Section~\ref{sec:cond_entropies}, respectively. Note also that our notation follows the one of~\cite{tomamichel_quantum_2016}. \\

We use the terms "non-negative operators" and "positive operators" to refer to linear, non-negative or positive operators on a Hilbert space, respectively. For simplicity, we consider only finite dimensional Hilbert spaces throughout this thesis. Therefore, non-negative operators and positive operators can always be viewed as positive semi-definite and positive definite matrices (over the complex numbers), respectively. \\
Throughout this thesis, taking the inverse of a non-negative operator $\rho$ should be viewed as taking the inverse evaluated only on the support of $\rho$.\\

Note also that we do not use a specific basis for the logarithm in this thesis. However, the exponential function should be considered as the reverse function of the chosen logarithm.\\

\vspace{5mm}

\noindent A list of abbreviations we use is available at Table~\ref{TABabbrev} and a comprehensive list of symbols can be found in Table~\ref{TABsymbols}. Note that the notation for matrices is also used for operators on Hilbert spaces in this thesis. This causes no confusion, because we work with finite dimensional Hilbert spaces only. 

\begin{table}[!ht]
\renewcommand{\arraystretch}{1.3}
\caption{List of abbreviations}
\label{TABabbrev}
\centering
\begin{tabular}{l l}
\hline
CPTP & Completely positive, trace-preserving (linear map)\\
POVM & Positive operator valued measure \\
DPI & Data-processing inequality [cf.~\eqref{eq:DPI}]\\
ALT & Araki-Lieb-Thirring (inequality) [cf. Theorem~\ref{thm:ALT}]\\
GT & Golden-Thompson (inequality) [cf. Theorem~\ref{thm:GT}]\\
cq & classical quantum \\
\hline
\end{tabular}
\end{table}

\begin{table}[!ht]
\renewcommand{\arraystretch}{1.3}
\caption{Notational conventions for mathematical expressions}
\label{TABsymbols}
\centering
\begin{tabular}{l l}
\hline
Operators on Hilbert spaces &\\ \hline
$\rho$, $\sigma$ & Typical elements of the set of non-negative operators\\
$\ker(\rho)$& Kernel of a non-negative operator $\rho$ \\
$\sigma \gg \rho$& $\ker( \sigma) \subseteq \ker(\rho)$ \\
$\cD(A)$& Set of density operators on a quantum system A, \\
&\quad i.e., non-negative operators $\rho$ with $\tr\rho=1$ \\
$\rho_A$& Density operator on a quantum sytem $A$ \\
$|A|$& Dimension of the Hilbert space $A$ \\ \hline
Matrices&\\ \hline
$\textnormal{Mat}(m,n)$ & Complex $m \times n$ matrices\\
$\textnormal{U}(n)$ & Unitary $n \times n$ matrices\\
$A^{*}$ & Conjugate transpose of a matrix $A \in \textnormal{Mat}(n,n)$ \\
$A\geqslant 0$&The matrix $A$ is positive semi-definite \\
$A> 0$&The matrix $A$ is positive definite \\
$A \#_{\alpha} B $&$= A^{\frac{1}{2}} \left(A^{-\frac{1}{2}}BA^{-\frac{1}{2}} \right)^{\alpha} A^{\frac{1}{2}}$ \quad (for $A,B>0$)\\
& \quad [$\alpha$-\emph{weighted geometric mean}]\\
$[A,B] $&$=AB-BA$ \quad [\emph{Commutator}]\\ \hline
Norms&\\ \hline
$|A|$&$= \sqrt{AA^{*}}$  for any $A \in \textnormal{Mat}(n,n)$\\
$\norm{\cdot}_p$&Schatten $p$-quasi-norm (cf. Section~\ref{sec:Schatten}) \\
$\normT{\cdot}$&Any unitarily invariant norm (cf. Definition~\ref{def:unitarily_inv_norm}) \\
\hline
\end{tabular}
\end{table}

%
%
%

\clearpage
%
%
\lhead[\fancyplain{\scshape Appendix}
{\scshape Appendix}]
{\fancyplain{\scshape \leftmark}
  {\scshape \leftmark}}
\rhead[\fancyplain{\scshape \leftmark}
{\scshape \leftmark}]
{\fancyplain{\scshape Appendix}
  {\scshape Appendix}}

%
%

%
%
%

\printbibliography[heading=bibintoc,title=References]

%
%
%
%
%
\end{appendix}
%
%
%
%
\end{document}